\documentclass[twocolumn,aps,preprintnumbers,showpacs,superscriptaddress]{revtex4-1}

\usepackage{amsmath,amssymb,amsthm,amsfonts}

\usepackage{mathtools}

\usepackage{algpseudocode}

\newcommand{\tp}{^{\mathsf{T}}}      

\newtheorem{lemma}{Lemma}
\newtheorem{prop}{Proposition}

\newtheorem{fact}{Fact}

\newcommand{\Lem}[1]{{Lemma~\ref{lem:#1}}}

\newcommand{\Prop}[1]{{Proposition~\ref{prop:#1}}}
\newcommand{\Sect}[1]{{Section~\ref{sect:#1}}}

\newcommand{\Fig}[1]{{Fig.~\ref{fig:#1}}}

\newcommand{\Eq}[1]{{Eq.~(\ref{eq:#1})}}

\newcommand{\be}{\begin{equation}}
\newcommand{\ee}{\end{equation}}
\newcommand{\ba}{\begin{array}}
\newcommand{\ea}{\end{array}}
\newcommand{\bea}{\begin{eqnarray}}
\newcommand{\eea}{\end{eqnarray}}

\newcommand{\calZ}{{\cal Z }}
\newcommand{\calL}{{\cal L }}
\newcommand{\calR}{{\cal R }}
\newcommand{\calN}{{\cal N }}
\newcommand{\calH}{{\cal H }}

\newcommand{\calG}{{\cal G }}

\newcommand{\calC}{{\cal C }}
\newcommand{\calA}{{\cal A }}
\newcommand{\calD}{{\cal D }}

\newcommand{\calP}{{\cal P }}

\newcommand{\CC}{\mathbb{C}}

\newcommand{\RR}{\mathbb{R}}

\newcommand{\mps}[1]{\mathrm{MPS}{({#1})}}
\newcommand{\mpo}[1]{\mathrm{MPO}{({#1})}}

\begin{document}

\title{Efficient Algorithms for Maximum Likelihood Decoding in the Surface Code}

\newcommand{\IBM}{IBM  Watson Research Center, Yorktown Heights, NY 10598, USA}

\author{Sergey \surname{Bravyi}}
\affiliation{\IBM}
\author{Martin \surname{Suchara}}
\affiliation{\IBM}
\author{Alexander \surname{Vargo}}
\affiliation{\IBM}

\date{\today}

\begin{abstract}
We describe two implementations of the optimal error correction algorithm known as the maximum likelihood decoder (MLD) for the 2D surface code with a noiseless syndrome extraction. First, we show how to  implement MLD exactly in time $O(n^2)$,
where $n$ is the number of code qubits. Our implementation uses a reduction from MLD to simulation of 
matchgate quantum circuits. This reduction  however requires a special noise model with independent bit-flip and phase-flip errors. Secondly, we show how to implement MLD approximately for more general noise models
using matrix product states (MPS).
Our implementation has running time
$O(n\chi^3)$ where  $\chi$ is a parameter that controls the approximation
precision. The key step of our algorithm, borrowed from the DMRG method, 
is  a subroutine for contracting a tensor network on the two-dimensional  grid. 
The subroutine uses MPS with a bond dimension $\chi$
to approximate the sequence of tensors arising in the course of contraction.  
We benchmark the MPS-based decoder against the standard 
minimum weight matching decoder observing a significant reduction
of the logical error probability  for $\chi\ge 4$.
\end{abstract}

\maketitle

\section{Introduction}


The surface code~\cite{Kitaev03,DKLP01} is one of the simplest and
most studied quantum error correcting codes. It 
can be realized on a two-dimensional grid of qubits such that the codespace is defined
by simple four-qubit parity check operators acting on nearest-neighbor qubits. 
Recent years
have witnessed a surge of interest in the surface code as a promising architecture for 
a scalable quantum computing~\cite{Fowler12surface,Ghosh12}.
Experimental advances in manufacturing of multi-qubit devices~\cite{Chow13,Barends14logic}
give us hope that a small-scale quantum memory based on the surface
code may become a reality soon. Given high operational costs of a quantum hardware
compared with the classical one, it is crucial
to put enough efforts in optimizing algorithmic, or software aspects of error correction. 
In the present paper we focus on optimizing the  decoding algorithm  that takes
as input measured syndromes of the parity checks  and computes a recovery 
operation returning a corrupted state of the memory back to the codespace.

As the name suggests, the maximum likelihood decoder (MLD)
is an algorithm that finds a recovery operation maximizing the probability of 
a successful error correction conditioned on the observed error syndrome. 
By definition, MLD is the optimal error correction algorithm for a fixed
quantum code and a fixed  noise model. The first  rigorous definition of MLD
for the surface codes was proposed by Dennis et al~\cite{DKLP01}. 
An important  observation made in~\cite{DKLP01} was that the computational
problem associated with MLD can be reduced to  computing the partition function 
of a classical Ising-like Hamiltonian on the two-dimensional lattice. 
This observation has generated a vast body of work exploring connections
between MLD and the statistical physics of disordered Ising-like Hamiltonians,
see for instance~\cite{Katzgraber09,Bombin10topological,Wootton12high,BombinPRX12}.
The insights made in~\cite{DKLP01} have also guided the search for
efficient implementations of MLD. Although an exact and efficient algorithm for MLD
remains an elusive goal, several approximate polynomial-time algorithms have
been discovered, most notably
the  renormalization group decoder due to
Duclos-Cianci and Poulin~\cite{Duclos10}, 
and the Markov chain Monte Carlo method
due to Hutter, Wootton, and Loss~\cite{Hutter14}.
In the case of concatenated codes an efficient exact algorithm for MLD based on the message
passing algorithm was proposed by Poulin~\cite{Poulin06optimal}.
By comparing MLD with the level-by-level decoder  commonly used for concatenated
codes, Ref.~\cite{Poulin06optimal} found that MLD offers a significant advantage with
almost two-fold increase of the error threshold for the depolarizing noise
and a significant reduction of the logical error probability. 

Here we propose an alternative method of implementing MLD 
in the case of the surface code
for two simple noise models known as the bit-flip noise and the depolarizing noise.
Our method combines the ideas of Dennis et al~\cite{DKLP01} and the standard
classical-to-quantum mapping  from classical 2D spin systems in the thermal equilibrium to
quantum 1D spin chains. 
It enables us to reduce the computational problem associated with MLD
to simulating a particular type of quantum dynamics for a chain of qubits.  

In the case of the bit-flip noise, MLD can be reduced to simulating
a quantum circuit with  a special type of two-qubit nearest-neighbor gates known as matchgates.
It was shown by Valiant~\cite{Valiant02} that quantum circuits composed
of matchgates can be efficiently simulated by classical means.
Matchgate circuits and their generalizations give rise to efficient holographic algorithms
for certain combinatorial problems~\cite{Valiant08} and efficient tensor network contraction 
methods~\cite{Cai06,Bravyi09}. Matchgate-based algorithms
have been  used to simulate quantum dynamics in systems of fermionic modes with quadratic
interactions~\cite{Knill01,TD02} and study statistics of dimer coverings
in classical lattice models~\cite{Fisher61,Kasteleyn61,Temperley61}. Here we demonstrate
that matchgates also have applications for quantum error correction. 
Our simulation algorithm based on fermionic Gaussian states~\cite{Bravyi04}
provides an exact implementation of MLD with the running time $O(n^2)$, where $n$
is the number of code qubits. The same algorithm can also be applied
to a noise model with independent bit-flip and phase-flip errors.
We note that a similar but technically different algorithm has been used by Merz and Chalker
in the numerical study of the random-bond 2D Ising model~\cite{Merz01}.

In the case of the depolarizing noise, MLD can be reduced to simulating
the dynamics generated by
matrix product operators with a small bond dimension. To perform the simulation  efficiently
we  conjecture that all intermediate states generated by this dynamics are weakly entangled. 
This enables us to employ a vast body of efficient classical algorithms for simulating
weakly entangled  quantum spin chains based on matrix product states (MPS),
see~\cite{Vidal03,Verstraete04,Murg07,Verstraete08review,Schollwock11}.
Our approximate implementation of MLD for the depolarizing noise has running time
$O(n\chi^3)$ where  $\chi$ is a parameter that controls the approximation
precision (the bond dimension of the MPS).  Although we do not have any rigorous arguments in support of the weak  entanglement conjecture,
it reflects the physical intuition that the classical 2D spin system associated with MLD
has a finite correlation length for error rates below the threshold~\cite{DKLP01}.
Accordingly, one should expect that the classical-to-quantum mapping  cannot generate
highly entangled states since the latter require long-range correlations. Furthermore,
we have  justified the conjecture numerically by applying the MPS-based decoder
to the bit-flip noise~\footnote{It should be emphasized that the MPS-based decoder
is applicable to any noise model that can be described by a stochastic i.i.d. Pauli noise.
In contrast, the decoder based on matchgates is only applicable to noise models
with independent bit-flip and phase-flip errors.}.
 We observed that the logical error probabilities of the exact MLD
and the MPS-based decoder  with a relatively small bond dimension $\chi=6$
are virtually indistinguishable.
Likewise, in the case of the depolarizing noise we observed that the 
logical error probability exhibits a fast convergence as a function of $\chi$ suggesting
that the MPS-based decoder with $\chi=6$ implements nearly exact MLD. 

Finally, we benchmark the exact and the approximate implementations of MLD
against the commonly studied minimum weight matching (MWM) decoder~\cite{DKLP01,Wang2011surface}.
The benchmarking was performed for a fixed code distance $d=25$ and a wide range 
of error rates.  In the case of the bit-flip noise we observed that the MWM decoder
approximates the logical error probability of MLD within a factor of two.
The observed difference between MLD and the MWM decoder can be attributed to the
fact that the latter ignores the error degeneracy~\cite{Stace10}.
Since the observed difference is relatively small, we conclude that ignoring the error degeneracy
does not have a significant impact on the decoder's performance for the studied noise model.
In the case of the depolarizing noise we observed that the MPS-based decoder
is far superior than the MWM decoder offering more than
two orders of magnitude reduction of the logical error probability
even for small values of $\chi$.
This can be attributed to the fact that  the MWM decoder often fails to  find the minimum weight error consistent with the syndrome since it  ignores correlations between $X$ and $Z$ errors~\cite{Fowler13optimal}.

The rest of the paper is organized as follows.
We formally define the maximum likelihood decoder, the studied noise models,
and the surface code in Sections~\ref{sect:mld},\ref{sect:noise}, and~\ref{sect:scodes} respectively. 
Our exact implementation of MLD for the bit-flip noise is described in \Sect{Xnoise}.
The approximate implementation of MLD  based on matrix product states
is presented  in \Sect{XYZnoise}.  
A comparison between the exact MLD, the approximate MLD with various bond
dimensions $\chi$, and the minimum weight matching decoder
is presented in \Sect{numerics} that describes our numerical results. 

\section{Maximum Likelihood Decoder}
\label{sect:mld}

In this section we formally define MLD.  We consider a quantum memory composed 
of $n$ physical qubits. Let $\calH=(\CC^2)^{\otimes n}$ be the full $n$-qubit Hilbert space
and $\calP$ be the group of $n$-qubit Pauli operators.
By definition, any element of $\calP$ has a form $f=c f_1\otimes f_2 \otimes \cdots \otimes f_n$,
where $f_j\in \{I,X,Y,Z\}$ are single-qubit Pauli operators and $c\in \{\pm 1,\pm i\}$ is an overall phase factor. 
 A quantum code of stabilizer
type is defined by an abelian {\em stabilizer group} $\calG\subset \calP$ such that $-I\notin \calG$.
Quantum codewords are $n$-qubit states invariant under the action of any element of $\calG$.
Such states define a {\em codespace} 
\[
\calH_0=\{ \psi \in \calH\, : \, g \psi=\psi \quad \mbox{for all $g\in \calG$}\}.
\]
The encoding step amounts to initializing the memory in some (unknown) state $\rho$
supported on the codespace $\calH_0$.

We shall consider a stochastic Pauli noise described by a linear map
\begin{equation}
\label{eq:final}
\calN(\rho)=\sum_{f\in \calP} \pi(f) \, f\rho f^\dag,
\end{equation}
where $\pi$ is some normalized probability distribution on the Pauli group. 
Since the initial state $\rho$ is supported on the codespace $\calH_0$,
one has $f\rho f^\dag = \rho$ for any $f\in \calG$. By the same token,
$f\rho f^\dag = h\rho h^\dag$ whenever  $f\calG =h\calG$.
Given a Pauli operator $f\in \calP$, a subset $f\calG\equiv \{ fg\, : \, g\in \calG\}$
is called a {\em coset}
of $\calG$.
Clearly, $\calP$ is a disjoint union of cosets
$\calC^\alpha=f^\alpha \calG$, where $f^\alpha$ is some  fixed representative of $\calC^\alpha$.
The above shows that 
errors in the same coset have the same action on the codespace.
Thus
\begin{equation}
\label{eq:final1}
\calN(\rho)=\sum_\alpha \pi(f^\alpha\calG) \cdot f^\alpha \rho f^\alpha,
\end{equation}
where the sum ranges over all cosets of $\calG$ and 
\[
\pi(f\calG)\equiv \sum_{g\in \calG} \pi(fg).
\]
For simplicity, here we assumed that all coset representatives $f^\alpha$
are hermitian operators.
We shall refer to the quantity $\pi(f\calG)$ as a {\em coset probability}.

At  the decoding step  one attempts to guess the coset of the stabilizer group
that contains the actual error based on a partial information about the error known
as a syndrome. More precisely, let $g^{1},\ldots,g^{m}\in \calG$ be some fixed set
of generators of $\calG$. Since the generators $g^{i}$ pairwise commute, they can
be diagonalized simultaneously. A configuration of eigenvalues  $g^{i}=\pm 1$
can be described by a {\em syndrome}
$s\in \{0,1\}^m$ such that $g^{i}=(-1)^{s_i}$ for all $i=1,\ldots,m$. 
Assuming that the generators $g^{i}$ are independent, there are $2^m$ possible syndromes. The full Hilbert space can be decomposed into a
direct sum of syndrome subspaces
\[
\calH=\bigoplus_{s\in \{0,1\}^m}\;  \calH_s,
\]
where $\calH_s=\{\psi\in \calH\, : \, g^{i}\psi=(-1)^{s_i} \psi \quad \mbox{for all $i$}\}$.
Note that the codespace $\calH_0$ corresponds to the zero syndrome. 
A Pauli operator $f\in \calP$ is said to have a syndrome $s$
iff $fg^{i}=(-1)^{s_i} g^{i} f$ for all $i=1,\ldots,m$. Equivalently, $f$ has a syndrome $s$ iff $f\calH_0=\calH_s$.
For each syndrome $s$ let us choose  some fixed Pauli operator 
$f(s)$ with the syndrome $s$.
One can easily check that the set of all Pauli operators with a syndrome $s$ coincides
with the coset $f(s)\calC(\calG)$, where 
\[
\calC(\calG)=\{ f\in \calP\, : \, fg=gf \quad \mbox{for all $g\in \calG$}\}
\]
is a group known as the centralizer of $\calG$.
Note that $\calG\subseteq \calC(\calG)$. Thus each coset of $\calC(\calG)$
can be partitioned into a disjoint union of several cosets of $\calG$. 
In the present paper we only consider stabilizer codes with a single logical qubit.
Let $\overline{X},\overline{Y},\overline{Z}\in \calC(\calG)\backslash\calG$ be the logical Pauli
operators on the encoded qubit. 
Then each coset of $\calC(\calG)$ consists of four disjoint cosets of $\calG$,
namely,
\begin{equation}
\label{eq:CG1}
f(s)\calC(\calG)=\calC_I^s \; \cup \; \calC_X^s \; \cup \; \calC_Y^s \; \cup \; \calC_Z^s,
\end{equation}
where
\begin{equation}
\label{eq:CG2}
\calC_I^s=f(s)\calG, \quad  \calC_X^s=f(s)\overline{X}\calG, 
\end{equation}
\begin{equation}
\label{eq:CG3}
\calC_Y^s=f(s)\overline{Y} \calG, \quad \mbox{and} \quad \calC_Z^s=f(s)\overline{Z} \calG.
\end{equation}

The decoding step starts by a {\em syndrome measurement}
that projects the corrupted state $\calN(\rho)$ onto one of the syndrome subspaces $\calH_s$.
The above arguments show that 
$f\rho f^\dag$ has support on $\calH_s$ 
iff $f\in f(s)\calC(\calG)$. Thus
the syndrome measurement reveals the coset of $\calC(\calG)$ that contains
the error $f$, whereas our goal is to determine which coset of $\calG$ contains $f$.
Using Eqs.~(\ref{eq:final1},\ref{eq:CG1}-\ref{eq:CG3}) one can write 
the post-measurement  (unnormalized) state as
\begin{equation}
\begin{array}{rcl}
\rho(s)&=&\pi(\calC_I^s)\cdot f(s) \rho f(s)\\
&+& \pi(\calC_X^s)\cdot f(s)\overline{X} \rho \overline{X} f(s) \\
&+& \pi(\calC_Y^s)\cdot f(s)\overline{Y} \rho \overline{Y} f(s) \\
&+& \pi(\calC_Z^s)\cdot f(s)\overline{Z} \rho \overline{Z} f(s), \\
\end{array}
\end{equation}
where $s$ is the observed syndrome.
Here we assumed for simplicity that $f(s)$ and the logical operators $\bar{X},\bar{Y},\bar{Z}$
are hermitian. This shows that the effective noise model
conditioned on the syndrome can be described by applying one of the four 
Pauli errors $f(s), f(s)\bar{X}, f(s)\bar{Y}$, and $f(s)\bar{Z}$ with 
probabilities $\pi(\calC_I^s), \pi(\calC_X^s), \pi(\calC_Y^s)$, and
$\pi(\calC_Z^s)$ respectively. Clearly, the best possible error correction
algorithm for this effective noise model is to choose a recovery
operator as  the most likely of the four
errors. Equivalently, we should choose a recovery operator 
as any Pauli operator that belongs to the most likely of
the four cosets $\calC_I^s, \calC_X^s,\calC_Y^s,\calC_Z^s$
which we denote $\calC_{\mathrm{ML}}^s$. 
These steps can be summarized as follows.

\begin{center}
\fbox{\parbox{0.95\linewidth}{
\begin{algorithmic}
\State{\hspace{2.2cm} {\bf ML  Decoder}}
\State{{\bf Input:} syndrome $s\in \{0,1\}^m$}
\State{{\bf Output:} recovery operator $g\in \calP$}
\State{}
\State{$f(s) \gets$  any Pauli operator with a syndrome $s$}
\State{$\calC_{\mathrm{ML}}^s \gets \mathrm{arg} \max_\calC \pi(\calC)$, where $\calC\in \{ \calC_I^s, \calC_X^s,\calC_Y^s,\calC_Z^s\}$}
\State{\Return any $g\in \calC_{\mathrm{ML}}^s$}
\end{algorithmic}
}}
\end{center}
The final step of the decoding is to apply the optimal recovery operator $g$.
It results in a state $g\rho(s)g^\dag$. 
We conclude that MLD correctly identifies the coset of $\calG$
that contains the actual error and  maps the corrupted state $\calN(\rho)$ back to the encoded
state $\rho$ with a probability
\[
P_{success}=\sum_{s\in \{0,1\}^m} \pi(\calC_{\mathrm{ML}}^s).
\]
In what follows we shall always ignore overall phase factors of Pauli operators.
Such phase factors are irrelevant for
our purposes since they do not change the outcome of error correction.

\section{Noise models}
\label{sect:noise}

We  shall consider a stochastic i.i.d. Pauli noise
\[
\calN=\bigotimes_{j=1}^n \calN_j,
\]
where
\[
\calN_j(\rho)=(1-\epsilon) \rho + \epsilon_X X\rho X + \epsilon_Y Y\rho Y + \epsilon_Z Z\rho Z
\]
and $\epsilon\equiv \epsilon_X+\epsilon_Y+\epsilon_Z$ is called an {\em error rate}. 
Two commonly studied noise models
are the classical bit-flip noise where only $X$-type errors are allowed (the X-noise)
and the depolarizing noise where all types of errors are equally likely.
The formal definitions are given below. 
\begin{eqnarray}
\mbox{\bf X-noise} &:&  \quad  \epsilon_X=\epsilon, \quad \epsilon_Y=\epsilon_Z=0, \nonumber \\
\mbox{\bf Depolarizing noise} &:& \quad 
\epsilon_X=\epsilon_Y=\epsilon_Z=\epsilon/3.\nonumber 
\end{eqnarray}
The corresponding probability distributions on the Pauli group are
\[
\pi(f)=(1-\epsilon)^{n-|f|} (\epsilon/3)^{|f|}
\]
for the depolarizing noise and 
\[
\pi(f)=\left\{
\ba{rcl}
(1-\epsilon)^{n-|f|} \epsilon^{|f|} &\mbox{if} & f\in \calP^X, \\
0 && \mbox{otherwise}\\
\ea
\right.
\]
for the X-noise. Here $|f|$ denotes the Hamming weight
of $f$, that is, the number of qubits on which $f$ acts nontrivialy,
while  $\calP^X\subset \calP$ denotes the subgroup generated by single-qubit Pauli $X$ operators.

One may also consider a noise model with independent bit-flip and phase-flip errors,
that is, $\epsilon_X=\epsilon_Z$ and $\epsilon_Y=(\epsilon_X)^2$. Since there are
no correlations between the two types of errors, one can perform error correction 
independently for bit-flip and phase-flip errors. Furthermore, since correcting phase-flip errors
is equivalent to correcting bit-flip errors on the surface code lattice rotated by $90^\circ$,
it suffices to consider the $X$-noise model only.

\section{Surface codes}
\label{sect:scodes}

We consider the surface code on a square lattice of size $d\times d$ with 
open boundary conditions. The boundaries parallel to the horizontal (vertical) axis are smooth (rough).
The surface code lattice with $d=3$ is shown on \Fig{d3}.
For the chosen geometry the surface code encodes one logical qubit
into $n=d^2+(d-1)^2$ physical qubits with the minimum distance $d$. 
We shall always consider odd values of $d$ such that the code corrects 
any combination of $(d-1)/2$ single qubit errors. 
%
%
Let $A_u$ and $B_p$ be the stabilizers of the surface code associated
with a site $u$ and a plaquette $p$ respectively. We have $B_p=\prod_{e\in p} X_e$,
where the product runs over all edges $e$ making up the boundary of $p$.
Likewise, $A_u=\prod_{e\ni u} Z_e$, where the product runs over all edges $e$ incident to $u$. 
Let $\calG^Z=\langle A_u\rangle$ and $\calG^X=\langle B_p\rangle$ be the subgroups of the Pauli group $\calP$ 
generated by all site stabilizers and all plaquette stabilizers respectively. 
Finally, let $\calG=\langle A_u, B_p\rangle$ be the full stabilizer group. 
Logical Pauli operators $\overline{X},\overline{Z}$ are shown on Fig.~\ref{fig:logical},
while  $\overline{Y}=i\overline{X}\overline{Z}$.

\begin{figure}[h]
\centerline{\includegraphics[height=6cm]{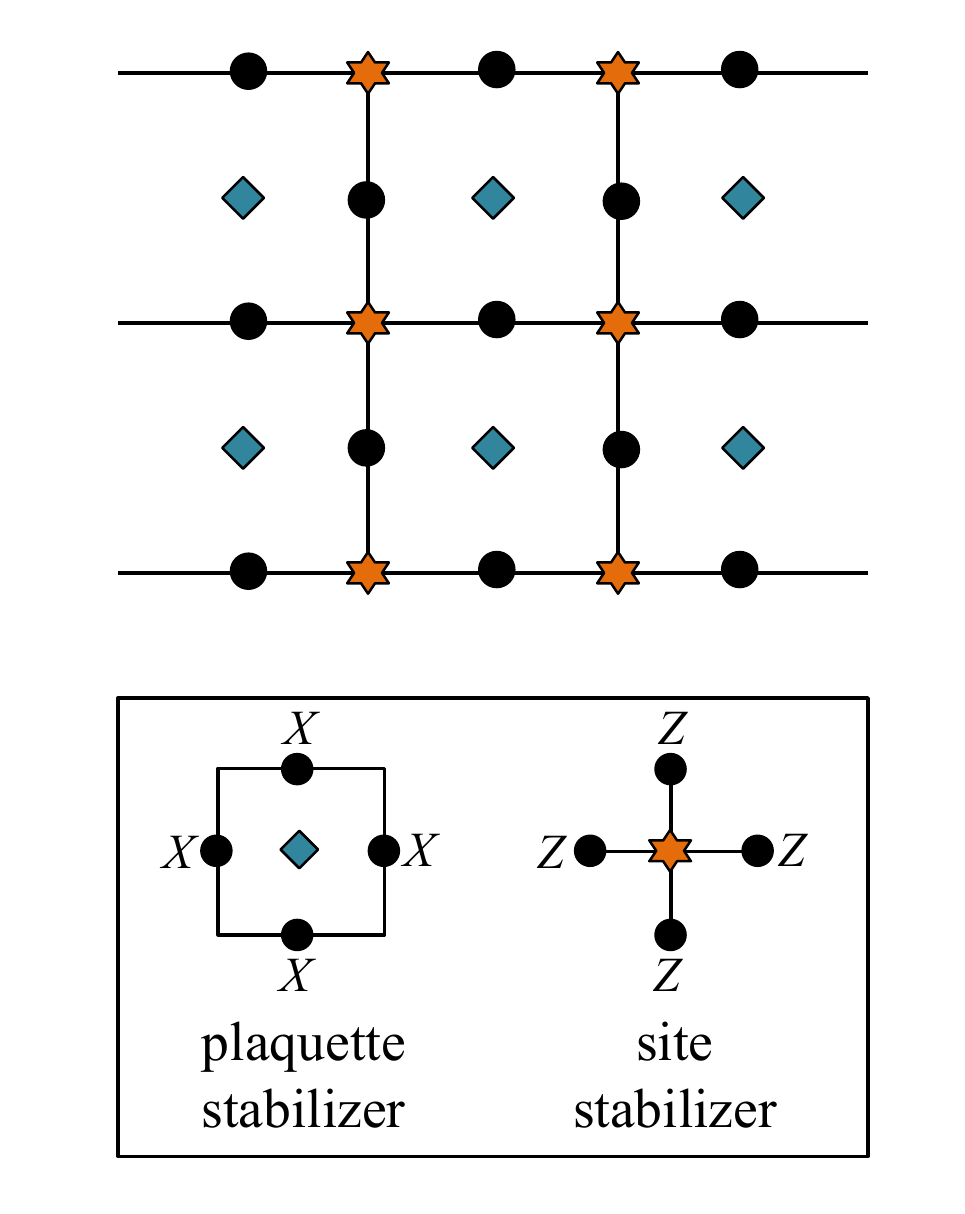}}
\caption{Distance-$3$ surface code.
Solid dots, stars, and diamonds 
indicate locations of qubits,  site stabilizers, and plaquette stabilizers respectively. 
Stabilizers located near the boundary act only on three qubits. 
The distance-$d$ surface code has $d^2$ qubits on horizontal edges, $(d-1)^2$ qubits on vertical edges,
and  $d(d-1)$ stabilizers of each type.
\label{fig:d3}
}
\end{figure}

\begin{figure}[h]
\centerline{\includegraphics[width=7cm]{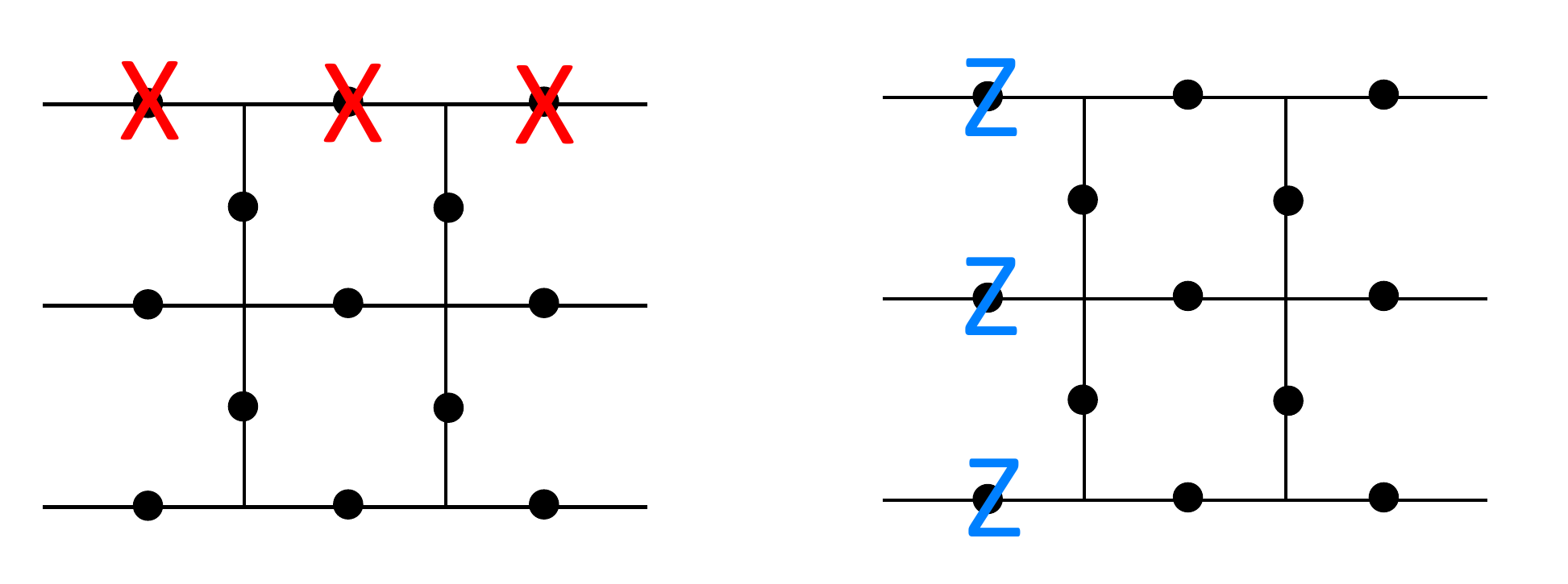}}
\caption{Logical Pauli operators $\overline{X}$ (left) and $\overline{Z}$ (right). 
\label{fig:logical}
}
\end{figure}

By a slight abuse of notations, below we shall often  identify a Pauli operator 
$f$  with the subset of edges in the lattice
on which $f$ acts non-trivially. 

\section{Exact algorithm}
\label{sect:Xnoise}

In this section we consider the X-noise and describe an exact implementation
of MLD. We begin by specializing MLD  to 
the X-noise (\Sect{MLD4X}) and describing our algorithm (\Sect{algorithm1}).
A reader interested only in the question of {\em how} 
the algorithm works can skip the remaining sections
explaining  {\em why} it works and proving its correctness.
 Specifically,  \Sect{reduction}  shows how
to express the coset probability  as  a matrix element of a matchgate quantum circuit.
Our derivation partially follows the one of Refs.~\cite{DKLP01,Merz01}.
An efficient method of simulating matchgate circuits based on fermionic Gaussian states
 is described in \Sect{gaussian}. The material of this section
 mostly follows Ref.~\cite{Bravyi04}.

\subsection{Specializing the ML decoder to X-noise}
\label{sect:MLD4X}

Let $s$ be the input syndrome and $f(s)\in \calP$ be some fixed Pauli error consistent with $s$.
We can always choose $f(s)\in \calP^X$, that is, such that $f(s)$ acts on any qubit by $I$ or $X$. Indeed,
since only $X$-type errors can appear with a non-zero probability,  the syndromes of all plaquette stabilizers must be zero.
Let $s_u$ be the syndrome of a site stabilizer $A_u$. We  choose the desired error $f(s)$ by connecting each site $u$
with a non-zero syndrome $s_u$ to the left boundary by a horizontal string of $X$ errors and adding all such strings
modulo two. Note that $f(s)$ can be constructed in time $O(n)$. 

Let $\pi$ be the probability distribution on the Pauli group describing the X-noise, see \Sect{noise}.
To implement the ML decoder it suffices to compute the four coset probabilities
$\pi(\calC_I^s), \pi(\calC_X^s), \pi(\calC_Y^s)$, and $\pi(\calC_Z^s)$ as defined in \Sect{mld}.
Note that  $\pi(\calC_Y^s)=\pi(\calC_Z^s)=0$ since any element of these two cosets
acts by Pauli $Z$ on at least $d$ qubits. 
Choose any logical operator $\overline{L}\in  \{\overline{I},\overline{X}\}$
and let $f\equiv f(s)\overline{L}$. From now on we shall assume that $f$ is fixed.
Since $Z$-type errors are not allowed, one has $\pi(f\calG)=\pi(f\calG^X)$.
Thus it suffices to compute the coset probability $\pi(f\calG^X)$.

\subsection{Algorithm for computing the coset probability}
\label{sect:algorithm1}


In this section we describe an algorithm that takes as input an $X$-type Pauli operator $f$
and outputs the  coset probability $\pi(f\calG^X)$.
The algorithm has running time $O(n^2)$.

Let us begin by introducing some notations. 
The sets of all horizontal and vertical edges
of the surface code lattice will be denoted $H$ and $V$ respectively. 
For the code of distance $d$ one has $|H|=d^2$ and 
$|V|=(d-1)^2$. We partition the set $H$ into columns of edges
such that 
\[
H=H^1 \cup H^2 \cup \ldots \cup H^d,
\]
where $H^j$ denotes the $j$-th leftmost column of horizontal edges, see \Fig{layers}. 
Edges of every column $H^j$ will be labeled by integers $1,\ldots,d$ 
starting from the top edge. 
Likewise,
\[
V=V^1\cup V^2 \cup \ldots \cup V^{d-1},
\]
where $V^j$ denotes the $j$-th leftmost column of vertical edges, see \Fig{layers}. 
Edges of every column $V^j$ will be labeled by integers $1,\ldots,d-1$ 
starting from the top edge. 
We shall refer to $H^j$ and $V^j$ as horizontal and vertical columns respectively. 

\begin{figure}[h]
\centerline{\includegraphics[height=5cm]{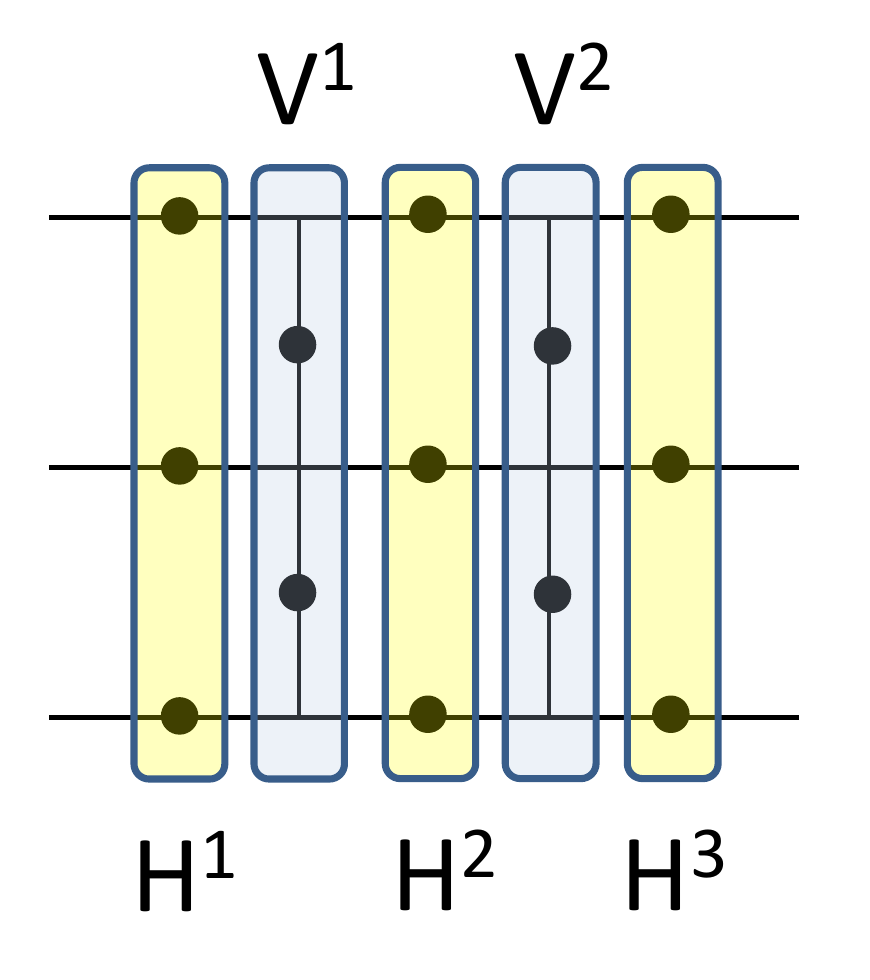}}
\caption{Partition of edges into `horizontal' columns  $H^1,\ldots,H^d$
and vertical columns $V^1,\ldots,V^{d-1}$. 
Every edge is identified with the respective code qubit (solid dot). 
\label{fig:layers}
}
\end{figure}

For each edge $e$ of the surface code lattice define a weight
\begin{equation}
\label{eq:wj}
w_e=\left\{ \ba{rcl}
\epsilon(1-\epsilon)^{-1} &\mbox{if} & e\notin f, \\
\epsilon^{-1}(1-\epsilon) &\mbox{if} & e\in f. \\
\ea
\right.
\end{equation}
Recall that $\epsilon$ is the error rate. 

For any integer $m\ge 1$ and a vector $\lambda\in \RR^m$
let $\calA(\lambda)$ be the anti-symmetric matrix of size $(m+1)\times (m+1)$
that contains $\lambda$ above the main diagonal and
$-\lambda$ below the main diagonal. For example, if $\lambda=(\lambda_1,\lambda_2,\lambda_3)$ then 
\[
\calA(\lambda)=\left[ \ba{cccc}
0 & \lambda_1 & 0 & 0 \\
-\lambda_1 & 0 & \lambda_2 & 0\\
0 & -\lambda_2 & 0 & \lambda_3 \\
0 & 0 & -\lambda_3 & 0 \\
\ea
\right].
\]
Let $\calD(\lambda)$ be the diagonal matrix of size $m\times m$ that contains $\lambda$
on the main diagonal. Define also a standard antisymmetric matrix
\begin{equation}
\label{eq:M0}
M_0=
\left[
\ba{ccccc}
0 &   &  &  &  1 \\
& {\left[ \ba{cc}  0 & 1\\ -1 & 0 \\ \ea \right]} & & & \\
& & \ddots & & \\
& & & {\left[ \ba{cc}  0 & 1\\ -1 & 0 \\ \ea \right]}  & \\
-1 & & & & 0 \\
\ea\right]
\end{equation}
such that $M_0$ has size $2d\times 2d$. The matrix $M_0$
contains $d-1$ blocks of size $2\times 2$ on the main diagonal
and two non-zero elements $M_{1,2d}=1=-M_{2d,1}$. 
All remaining elements of $M_0$ are zero. 
Let $I$ be the identity matrix of size $2d\times 2d$.

The first step of our algorithm  is to compute the
probability  of the input error  $\pi(f)=(1-\epsilon)^{n-|f|} \epsilon^{|f|}$
and the coefficients $w_e$ defined in \Eq{wj}. 
This step takes time $O(n)$. 
At each subsequent step  of the algorithm we maintain a pair $(M,\Gamma)$,
where $M$ is an antisymmetric real matrix of size $2d\times 2d$
and $\Gamma\ge 0$ is a real number. 
The algorithm calls two functions SimulateHorizontal$(j,M,\Gamma)$ and SimulateVertical$(j,M,\Gamma)$
that update  the pair  $(M,\Gamma)$ by applying a simple combination of matrix inversions and matrix multiplications. 

\begin{center}
\fbox{\parbox{0.9\linewidth}{
\begin{algorithmic}
\State{\hspace{2.2cm} {\bf Algorithm~1}}
\State{{\bf Input:} $X$-type Pauli operator $f$}
\State{{\bf Output:} Coset probability $\pi(f\calG^X)$}
\State{}
\State{Compute the coefficients $w_e$ defined in \Eq{wj}}
\State{$\pi(f)\gets (1-\epsilon)^{n-|f|} \epsilon^{|f|}$}
\State{$M\gets M_0$}
\State{$\Gamma\gets 2^{d-1}$}
\For{$j=1$ to $d-1$}
\State{\Call{SimulateHorizontal}{$j,M,\Gamma$}}
\State{\Call{SimulateVertical}{$j,M,\Gamma$}}
\EndFor
\State{\Call{SimulateHorizontal}{$d,M,\Gamma$}}
\State{\Return $\pi(f) \sqrt{\Gamma/2}  \cdot \det{(M+M_0)}^{1/4}$}
\end{algorithmic}
}}
\end{center}

\begin{center}
\fbox{\parbox{0.9\linewidth}{
\begin{algorithmic}
\Function{SimulateHorizontal}{$j,M,\Gamma$} 
\For{$i=1$ to $d$}
\State{$e\gets \mbox{$i$-th edge of the column $H^j$}$}
\State{$\Gamma \gets \Gamma \cdot (1+w_e^2)/2$}
\State{$t_i\gets (1-w_e^2)/(1+w_e^2)$}
\State{$s_i\gets 2w_e/(1+w_e^2)$}
\EndFor
\State{$A\gets \calA(t_1 0 t_2 0\ldots t_{d-1}0t_d)$}
\State{$B\gets \calD(s_1 s_1 s_2 s_2 \ldots s_d s_d)$}
\State{$\Gamma\gets \Gamma\cdot \sqrt{\det{(M+A)}}$}
\State{$M\gets A-B(M+A)^{-1} B$}
 \EndFunction
\end{algorithmic}
}}
\end{center}

\begin{center}
\fbox{\parbox{0.9\linewidth}{
\begin{algorithmic}
\Function{SimulateVertical}{$j,M,\Gamma$} 
\For{$i=1$ to $d-1$}
\State{$e\gets \mbox{$i$-th edge of the column $V^j$}$}
\State{$\Gamma \gets \Gamma \cdot (1+w_e^2)$}
\State{$t_i\gets 2w_e/(w_e^2+1)$}
\State{$s_i\gets (1-w_e^2)/(1+w_e^2)$}
\EndFor
\State{$A\gets \calA(0t_1 0 t_2 \ldots 0 t_{d-1} 0)$}
\State{$B\gets \calD(1 s_1 s_1 s_2 s_2 \ldots s_{d-1} s_{d-1} 1)$}
\State{$\Gamma\gets \Gamma\cdot \sqrt{\det{(M+A)}}$}
\State{$M\gets A-B(M+A)^{-1} B$}
 \EndFunction
\end{algorithmic}
}}
\end{center}

If implemented naively, each matrix inversion and each matrix multiplication
takes time $O(d^3)$. Likewise, computing each determinant takes time
$O(d^3)$. Simple counting then shows that the overall running time of the  algorithm 
is  $O(d^4)=O(n^2)$. Suggestions on improving  stability
of the algorithm against rounding errors can be found in \Sect{numerics}.

\subsection{Reduction to a matchgate quantum circuit}
\label{sect:reduction}
 
Consider any stabilizer $g\in \calG^X$. A simple algebra shows that
\[
\pi(fg)=\pi(f) \prod_{e\in g} w_e,
\]
where $w_e$ are the weights defined in \Eq{wj}. Thus
\[
\pi(f\calG^X)=\pi(f)\calZ(w),
\]
where $w=\{w_e\}$ is the list of coefficients $w_e$ and 
\begin{equation}
\label{Zloops}
\calZ(w)=\sum_{g\in \calG^X} \; \prod_{e\in g} w_e.
\end{equation}

Since the factor $\pi(f)$ is easy to compute, below we concentrate on 
computing  $\calZ(w)$. We shall express $\calZ(w)$
as a matrix element of a certain quantum circuit acting on $d$ qubits.
The circuit will be composed of single-qubit and two-qubit gates 
\begin{equation}
\label{eq:gates}
G(w)  \equiv \left[ \ba{cc}
1 & 0\\
0 & w \\
\ea
\right] \quad \mbox{and}\quad 
G'(w)\equiv  \left[ \ba{cccc}
1 & 0 & 0 & w \\
0 & 1 & w & 0 \\
0 & w & 1 & 0\\
w & 0 & 0 & 1\\
\ea
\right]
\end{equation}
where $w$ is a real parameter. We note that $G(w)$ and $G'(w)$ are not unitary gates. 
Let $\calH_d=(\CC^2)^{\otimes d}$ be the Hilbert space of $d$ qubits.
For each horizontal column $H^j$ and each vertical column $V^j$ defined 
at \Fig{layers} define linear operators $\hat{H}^j,\hat{V}^j$ acting on $\calH_d$ such that
\begin{equation}
\label{eq:hatH}
\hat{H}^j = G(w_{e_1})\otimes \cdots \otimes G(w_{e_d})
\end{equation}
and
\begin{equation}
\label{eq:hatV}
\hat{V}^j=
G_{12}'(w_{e_1}) G_{23}'(w_{e_2})  \cdots G_{d-1,d}'(w_{e_{d-1}}).
\end{equation}
Here the subscripts indicate the qubits acted upon by each gate and 
$e_i$ denotes the $i$-th  edge of the respective columns
$H^j$ and $V^j$ counting from the top to the bottom.
Finally, define a state 
\begin{equation}
\label{eq:psi}
|\psi_e\rangle=\sum_{x\in \{0,1\}^d_{even}}\; |x\rangle,
\end{equation}
where $\{0,1\}^d_{even}$ is the set of all $d$-bit binary strings with the even
Hamming weight. 
\begin{lemma}
\label{lem:mgates}
One has
\begin{equation}
\label{eq:Zw1}
\calZ(w)=\langle \psi_e | \hat{U} |\psi_e\rangle, 
\end{equation}
where
\begin{equation}
\label{eq:hatU}
\hat{U}=\hat{H}^d \hat{V}^{d-1} \cdots   \cdots\hat{H}^2  \hat{V}^{1} \hat{H}^1
\end{equation}
is a quantum circuit on $d$ qubits shown at \Fig{mgates}. 
\end{lemma}

\begin{figure}[h]
\centerline{\includegraphics[height=6cm]{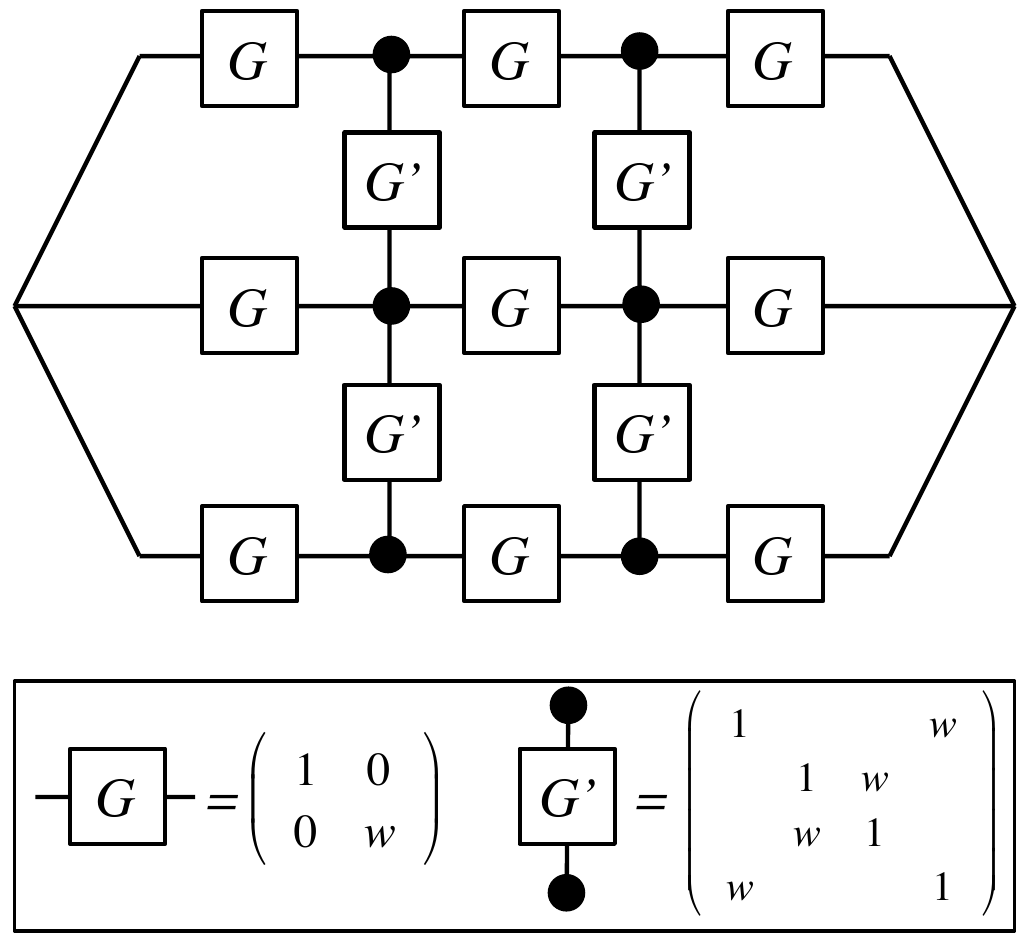}}
\caption{Computing the coset probabilities for the $X$-noise model
is equivalent to computing the matrix element $\langle \psi_e|\hat{U}|\psi_e\rangle$,
where $\hat{U}$ is a quantum circuit on $d$ qubits shown above  and $\psi_e$ is the superposition
of all even-weight $d$-bit strings. The above example is  for $d=3$. 
Each gate depends on a parameter $w_e$ defined in \Eq{wj}.
\label{fig:mgates}
}
\end{figure}
The gates $G(w)$ and $G'(w)$ defined in \Eq{gates} 
are examples of the so-called matchgates
discovered by Valiant~\cite{Valiant02}. It was shown in~\cite{Valiant02} that quantum circuits composed
of matchgates can be efficiently simulated by classical means.
In the next section we describe an alternative  algorithm
for computing the quantity $\langle \psi_e | \hat{U} |\psi_e\rangle$ 
based on fermionic Gaussian states 
with a running time $O(n^2)$. 
(For comparison, the original algorithm
of Ref.~\cite{Valiant02} would have running  time 
$O(n^3)$ since it requires computing the Pfaffian of a matrix of size $O(n)$.)

In the rest of this section we prove \Lem{mgates}. 
\begin{proof}
\begin{prop}
\label{prop:trivial}
For any  subset $T\subseteq H$ such that $|T\cap H^j|$ is even
for all $j=1,\ldots,d$ there
exists a unique $g\in \calG^X$ such that $g\cap H=T$.
\end{prop} 
\begin{proof}
Recall that a subset of edges $g$ is called a cycle iff any site
has even number of incident edges from $g$.
Let us first show that for any $T\subseteq H$ there exists exactly one
cycle $g$ such that $g\cap H=T$. 
Indeed, consider any vertical column $V^j$. It comprises
a set of sites $u_1,\ldots,u_d$ and a set of edges $e_1,\ldots,e_{d-1}$ (listed in the order
from the top to the bottom). Since $g\cap H^1=T\cap H^1$ and $g\cap H^2=T\cap H^2$,
the cycle condition at $u_1$ uniquely determines $g_{e_1}$. 
Once $g_{e_1}$ is determined, the cycle condition  $u_2$
uniquely determines $g_{e_2}$.  Continuing in this fashion uniquely  determines
$g\cap V^j$. Since $V^j$ can be any vertical column, we conclude that $g$
is uniquely determined by $T$. 
It remains to note that  $\calG^X$ coincides with the set of cycles that 
 have even intersection with any column $H^j$.
\end{proof}
Let $g(T)\in \calG^X$ be the Pauli operator constructed in \Prop{trivial}.
Then
\begin{equation}
\label{eq:Zloops2}
\calZ(w)=\sum_{T\subseteq H} \; \prod_{e\in g(T)} w_e,
\end{equation}
where the sum ranges over all subsets $T$ such that $|T\cap H^j|$ is even for all $j$.
Let $T^j\equiv T\cap H^j$. We can regard $T^j$ as a binary $d$-bit string
such that $T^j_i=1$ iff the $i$-th edge of $H^j$ belongs to $T$. 
Let  $|T^j\rangle\in \calH_d$ be the basis vector corresponding to $T^j$. 
Since $g(T)\cap H^j=T^j$, we have 
\begin{equation}
\label{eq:Zloops3}
\prod_{e\in g(T)\cap H^j} w_e = \langle T^j| G(w_{e_1})\otimes \cdots \otimes G(w_{e_d})|T^j\rangle,
\end{equation}
where $e_1,\ldots,e_d$ are the edges comprising the column $H^j$ listed in the order from the top
to the bottom and $G(w)$ is the single-qubit gate defined in \Eq{gates}. 

Consider now some vertical column $V^j$. Let
$e_1,\ldots,e_{d-1}$ be the edges comprising $V^j$ listed in the order from the top
to the bottom.
 We claim that 
\begin{eqnarray}
\label{eq:Zloops4}
\prod_{e\in g(T)\cap V^j} w_e &=& \langle T^j| G_{12}'(w_{e_1}) G_{23}'(w_{e_2}) \cdots \nonumber \\
&& \cdots G_{d-1,d}'(w_{e_{d-1}})|T^{j+1}\rangle,
\end{eqnarray}
where $G'(w)$ is the two-qubit gate defined in \Eq{gates} and $G_{i,i+1}'(w)\, : \, \calH_d\to \calH_d$ denotes the gate $G'(w)$ applied to the 
pair of qubits $i,i+1$. One can easily check \Eq{Zloops4} by
noting that $G'(w)=I\otimes I + wX\otimes X$ and
following the arguments given in proof of \Prop{trivial}
to reconstruct $g(T)\cap V^j$ from $T^j$ and $T^{j+1}$. 

Let $\{0,1\}^d_{\mathrm{even}}$ be the set of all $d$-bit strings
with even Hamming weight. Combining Eqs.~(\ref{eq:Zloops2},\ref{eq:Zloops3},\ref{eq:Zloops4})
one arrives at
\begin{eqnarray}
\label{eq:Zloops5}
\calZ(w)&=&\sum_{T^1,\ldots,T^d\in \{0,1\}^d_{\mathrm{even}} } \; \; 
\langle T^d|\hat{H}^d |T^d\rangle \langle T^d|\hat{V}^{d-1}|T^{d-1}\rangle \cdots \nonumber  \\
&& \cdots \langle T^{2} |\hat{V}^{1} |T^1\rangle \langle T^1|\hat{H}^1 |T^1\rangle.
\end{eqnarray}
where $\hat{H}^j$ and $\hat{V}^j$ are the linear operators on $\calH_d$
defined in Eqs.~(\ref{eq:hatH},\ref{eq:hatV}). 
 
Let $\calH_d^{even}\subseteq \calH_d$ be the subspace spanned by vectors
$|x\rangle$ with $x\in \{0,1\}^d_{\mathrm{even}}$. Note that the operators
$\hat{H}^j$ and $\hat{V}^j$ preserve $\calH_d^{even}$ since the gates
$G(w)$ and $G'(w)$ preserve the Hamming weight modulo two. 
The above observations imply that $\calZ(w)=\langle \psi_e | \hat{U} |\psi_e\rangle$,
which completes the proof of \Lem{mgates}.
\end{proof}

\subsection{Fermionic Gaussian states}
\label{sect:gaussian}

Let $\calH_d$ be the Hilbert space of $d$ qubits. 
For each $p=1,\ldots,2d$ define a {\em Majorana operator} $\hat{c}_p$ acting on $\calH_d$
such that 
\begin{equation}
\label{eq:majorana}
\hat{c}_{2j-1}=Z_1\cdots Z_{j-1} X_j \quad \mbox{and} \quad \hat{c}_{2j}=Z_1\cdots Z_{j-1} Y_j.
\end{equation}
The Majorana operators  obey the well-known commutation rules
\begin{equation}
\label{eq:CR}
\hat{c}_p \hat{c}_q + \hat{c}_q \hat{c}_p = 2I\delta_{p,q}, \quad \hat{c}_p^2=I, \quad \hat{c}_p^\dag =\hat{c}_p.
\end{equation}
We shall often use a formula
\begin{equation}
\label{eq:useful}
Z_j=(-i)\hat{c}_{2j-1}\hat{c}_{2j} \quad \mbox{and} \quad X_j X_{j+1} =(-i) \hat{c}_{2j}\hat{c}_{2j+1}.
\end{equation}

A  {\em covariance matrix} of a pure (unnormalized) state $\psi\in \calH_d$ is a $2d\times 2d$
matrix $M$ with matrix elements
\begin{equation}
\label{eq:M}
M_{p,q}=\frac{(-i)}{2 \langle \psi|\psi\rangle} \langle \psi| \hat{c}_p \hat{c}_q - \hat{c}_q \hat{c}_p |\psi\rangle.
\end{equation}
From \Eq{CR} one can easily check that $M$ is a real anti-symmetric matrix.

Consider as an example the state $\psi_e$ defined in \Eq{psi}.
Let us compute its covariance matrix $M$.  One can easily check that $\psi_e$
is a stabilizer state with the stabilizer group 
\[
\calG(\psi_e)=\langle X_1 X_2, X_2 X_3, \ldots, X_{d-1}X_d,  Z_1 Z_2 \cdots Z_d\rangle.
\]
Applying \Eq{useful} one can get an alternative set of generators 
that are quadratic in Majorana operators, 
\begin{equation}
\label{eq:gens}
\calG(\psi_e)=\langle (-i) \hat{c}_{2} \hat{c}_{3}, \ldots,  (-i)\hat{c}_{2d-2} \hat{c}_{2d-1}, 
(-i)\hat{c}_1\hat{c}_{2d}\rangle.
\end{equation}
This shows that $M_{2j,2j+1}=1$ for all $j=1,\ldots,d-1$ and $M_{1,2d}=1$. 
Furthermore, $\langle \psi_e|\hat{c}_p \hat{c}_q|\psi_e\rangle=0$
whenever $\hat{c}_p \hat{c}_q$ anti-commutes with at least one of the
generators defined in \Eq{gens}.
Combining the above observations one can easily check that $M=M_0$,
where $M_0$ is the standard anti-symmetric matrix defined in \Eq{M0}.

A state $\psi\in \calH_d$ is said to obey the Wick's theorem
iff the expectation value of any even tuple of Majorana operators
on $\psi$ can be computed from its covariance matrix $M$ using the formula
\begin{equation}
\label{eq:wick}
\langle \psi|i^m \hat{c}_{p_1} \hat{c}_{p_2} \cdots \hat{c}_{p_{2m}} |\psi\rangle
=\Gamma \cdot \mathrm{Pf}(M|_{p_1,p_2,\ldots,p_{2m}}),
\end{equation}
where  $\Gamma=\langle \psi|\psi\rangle$ is the norm of $\psi$,
$M|_{p_1,p_2,\ldots,p_{2m}}$ is the $2m\times 2m$ submatrix of $M$
formed by the rows and columns $p_1,p_2,\ldots,p_{2m}$, and $\mathrm{Pf}$ is 
the Pfaffian~\cite{Valiant02}. Recall that the Pfaffian of an anti-symmetric
matrix $K$  of size $2m\times 2m$ is defined as 
\[
\mathrm{Pf}(K)=\frac1{2^m m!} \, \calA(K_{1,2} K_{3,4} \cdots K_{2m-1,2m}),
\]
where $\calA$ stands for the anti-symmetrization over all $(2m)!$ permutations of indexes. 
For example, $\mathrm{Pf}(K)=K_{1,2}$ for $m=1$ and
\[
\mathrm{Pf}(K)=K_{1,2}K_{3,4}- K_{1,3} K_{2,4} + K_{1,2} K_{3,4}
\]
for $m=2$. 

A state $\psi\in \calH_d$ is called a (fermionic) {\em Gaussian state} 
iff it obeys the Wick's theorem and, in addition, all odd tuples of Majorana
operators have zero expectation value on $\psi$.
By definition, a Gaussian state $\psi$ is fully specified by the pair $(M,\Gamma)$,
where $M$ is the covariance matrix of $\psi$ and $\Gamma=\langle\psi|\psi\rangle$ is the norm.
Below we shall identify a Gaussian state 
and the corresponding pair  $(M,\Gamma)$.

We shall need the following well-known facts, see for instance Ref.~\cite{Bravyi04}.
\begin{fact}
A state $\psi$ is Gaussian iff its 
covariance matrix  obeys $MM^T=I$.
\end{fact}
One can easily check that standard anti-symmetric matrix $M_0$ defined in \Eq{M0} satisfies
$M_0M_0^T=I$. This shows that $\psi_e$ is a Gaussian state with the 
covariance matrix $M_0$ and the norm $\Gamma=2^{d-1}$. 
\begin{fact}
Let $\psi=(M,\Gamma)$ and $\phi=(M',\Gamma')$ be Gaussian states
of $d$ qubits. 
Then 
\begin{equation}
\label{eq:overlap}
|\langle \phi|\psi\rangle|=\frac{\sqrt{\Gamma \Gamma'}}{2^{d/2}} \, \det{(M+M')}^{1/4}.
\end{equation}
\end{fact}
\begin{fact}
Let $G$ be a (complex) anti-symmetric matrix of size $2d\times 2d$. 
Consider an operator 
\begin{equation}
\label{eq:W}
W=\exp{(\hat{G})}, \quad \hat{G}=\sum_{1\le p<q\le 2d} G_{p,q} \hat{c}_p \hat{c}_q.
\end{equation}
Then $W$ maps Gaussian states to Gaussian states. 
\end{fact}
The last fact will be very important for us since the operators
$\hat{H}^j$ and $\hat{V}^j$ constructed in \Sect{reduction} have the form \Eq{W}.
Indeed, consider the gates $G(w)_a$ and $G'(w)_{a,a+1}$
where $G(w),G'(w)$ are defined in \Eq{gates}
and 
the subscripts indicate which qubits are acted upon by the gate.
One can easily check that $G(w)_a=\sqrt{w}e^{\beta Z_a}$, where $\beta$
is defined through $e^{-2\beta}=w$. Likewise, $G'(w)_{a,a+1}=\sqrt{w} e^{\beta X_a X_{a+1}}$.
From \Eq{useful}  and  \Eq{hatH} one gets 
\begin{equation}
\label{eq:hatH1}
\hat{H}^j =\sqrt{w_{e_1} \cdots w_{e_d}} \exp{\left(
\sum_{a=1}^d \beta_a (-i)\hat{c}_{2a-1} \hat{c}_{2a} \right)},
\end{equation}
where $\beta_a$ is defined through
\[
e^{-2\beta_a}=w_{e_a}, \quad a=1,\ldots,d.
\]
Likewise, \Eq{hatV} implies
\begin{equation}
\label{eq:hatV1}
\hat{V}^j =\sqrt{w_{e_1} \cdots w_{e_{d-1}}} \exp{\left(
\sum_{a=1}^{d-1} \beta_a(-i)\hat{c}_{2a}\hat{c}_{2a+1}
\right)}.
\end{equation}
Since $\psi_e$ is a Gaussian state, 
Fact~3  implies that all intermediate states obtained from $\psi_e$
by applying the operators $\hat{H}^j$ and $\hat{V}^j$ are Gaussian. 
Therefore $\calZ(w)=\langle\psi_e|\hat{H}^d \hat{V}^{d-1} \cdots \hat{V}^{1} \hat{H}^1|\psi_e\rangle$
can be efficiently computed
if we have a  rule describing how the covariance matrix
and the norm of a Gaussian state change upon application of $\hat{H}^j$ and $\hat{V}^j$. 
The desired rule can be obtained using  a fermionic version of the Jamiolkowski duality 
between states and linear maps introduced in~\cite{Bravyi04}.
Let us first define a fermionic version of the maximally entangled state
for a bipartite system of $d+d$ qubits. Let $\hat{c}_1,\ldots,\hat{c}_{4d}$
be the Majorana operators defined for a system of $2d$ qubits
according to \Eq{majorana}. 
Define a $2d$-qubit state normalized $\psi_I$ such that
$\psi_I$ has a stabilizer group
\[
\calG(\psi_I)=\langle (-i) \hat{c}_a\hat{c}_{a+2d}, \quad a=1,\ldots,2d\rangle.
\]
One can easily check that $\psi_I$ has a covariance matrix
\[
M_I=\left[ \ba{cc} 0 & I \\ -I & 0 \\ \ea \right],
\]
where each block has dimensions $2d\times 2d$. Fact~1 implies that
$\psi_I$ is a Gaussian state. 
Let $W=\exp{(\hat{G})}$ be the operator defined in \Eq{W}.
Define a $2d$-qubit state
\begin{equation}
\label{eq}
\psi_W= (W\otimes I) \psi_I  \in  \calH_{2d}.
\end{equation}
Note that $\psi_W$ is a Gaussian state due to Fact~3.
Let 
\[
M_W=\left[ \ba{cc} A & B \\ -B^T & D\\ \ea \right]
\]
be the covariance matrix of $\psi_W$. Here $A,B,D$ are some 
matrices of size $2d\times2d$.  Let $\Gamma_W=\langle \psi_W|\psi_W\rangle$.  
We shall need the following fact proved in~\cite{Bravyi04}.
\begin{fact}
Let $\psi=(M,\Gamma)$ be a Gaussian state. Then $W\psi=(M',\Gamma')$, where
\begin{equation}
\label{eq:map1}
M'=A-B(M-D)^{-1} B^T 
\end{equation}
and 
\begin{equation}
\label{eq:map2}
\Gamma'=\Gamma_W \Gamma \sqrt{\det{(M-D)}}.
\end{equation}
\end{fact}
It remains to compute $(M_W,\Gamma_W)$ for the two special 
cases $W=\hat{H}^j$ and $W=\hat{V}^j$. 

We shall perform the calculation for $W=\hat{H}^j$ since both cases
are quite similar. 
First we note that $W$ is a product of operators acting on disjoint
pairs of Majorana modes $(\hat{c}_{2a-1},\hat{c}_{2a})$. 
Accordingly, $\psi_W$ is a product of states involving disjoint  $4$-tuples of  Majorana modes
$(\hat{c}_{2a-1},\hat{c}_{2a},\hat{c}_{2a-1+2d},\hat{c}_{2a+2d})$.
It suffices to compute the covariance matrix and the norm for each of those $4$-tuples.
Equivalently, it suffices to do the calculation for $d=1$. In this case
$W=G(w)$ is the single-qubit operator defined in \Eq{gates}.
By definition, $\psi_I$ is a two-qubit state with stabilizers
$(-i)\hat{c}_1\hat{c}_3=-Y_1X_2$ and $(-i)\hat{c}_2\hat{c}_4=X_1Y_2$. 
It can be written explicitly as 
\[
|\psi_I\rangle = \frac1{\sqrt{2}} (|10\rangle+  i|01\rangle).
\]
Hence
\[
|\psi_W\rangle\equiv (W\otimes I)|\psi_I\rangle = \frac1{\sqrt{2}} (w|10\rangle + i|01\rangle).
\]
This state has norm
\[
\Gamma_W=\langle\psi_W|\psi_W\rangle = \frac12 (1+w^2).
\]
To compute the covariance matrix $M_W$ we shall use a shorthand notation
\[
\langle \cdot \rangle \equiv \frac{\langle \psi_W|\cdot |\psi_W\rangle}{\langle\psi_W|\psi_W\rangle}.
\]
By definition, 
\[
(M_W)_{p,q}=\langle (-i)\hat{c}_p \hat{c}_q \rangle \quad \mbox{for} \quad 1\le p<q\le 4.
\]
A straightforward calculation shows that the only non-zero elements (with $p<q$)  of 
$M_W$ are 
\[
(M_W)_{1,2}= \langle Z_1 \rangle = \frac{1-w^2}{1+w^2} \equiv t,
\]
\[
(M_W)_{1,3}= \langle -Y_1 X_2 \rangle =\frac{2w}{1+w^2}\equiv s,
\]
\[
(M_W)_{2,4}=\langle X_1 Y_2\rangle =s
\]
and
\[
(M_W)_{3,4}=\langle Z_2\rangle = -t.
\]
Thus 
\begin{equation}
\label{eq:MW}
M_W=\left[ \ba{cc} A & B \\ -B^T & D \\ \ea \right],
\end{equation}
where
\[
A=-D=\left[ \ba{cc} 0 & t \\ -t & 0 \\ \ea \right] \quad \mbox{and} \quad
B=\left[ \ba{cc} s & 0 \\ 0 & s \\ \ea \right].
\]
For an arbitrary $d$ we just need to take a direct sum of $d$ matrices $M_W$
as above  and take the product of $d$ normalizing coefficients $\Gamma_W$
defined above. This yields 
\[
\Gamma_W=\prod_{a=1}^d \frac12(1+w_{e_a}^2)
\] 
whereas $M_W$ is given by \Eq{MW} where 
$A=-D=\calA(t_1,0,t_2,0,\ldots,0,t_d)$ and $B=\calD(s_1, s_1,\ldots,s_d,s_d)$ with
\[
t_a=\frac{1-w_{e_a}^2}{1+w_{e_a}^2} \quad \mbox{and} \quad
s_a=\frac{2w_{e_a}}{1+w_{e_a}^2}.
\]
Combining the above analysis and Fact~4 we infer   that the function SimulateHorizontal$(j,M,\Gamma)$
defined in \Sect{algorithm1} describes how the covariance matrix
and the norm of a Gaussian state change under application of 
the operator $\hat{H}^j$. 
A similar calculation shows that the function SimulateVertical$(j,M,\Gamma)$
 describes how the covariance matrix
and the norm of a Gaussian state change under application of 
the operator $\hat{V}^j$. The very last step of Algorithm~1
correspond to computing the overlap between
$\psi_e$ and the final state 
$\hat{H}^d \cdots \hat{V}^{1} \hat{H}^1 \psi_e$
using \Eq{overlap}.
This completes the proof of correctness of Algorithm~1.

\section{Approximate algorithm}
\label{sect:XYZnoise}

In this section we  describe an approximate algorithm
for computing the coset probabilities. It is applicable to a general stochastic i.i.d. Pauli
noise including the depolarizing noise. 
We assume some level of familiarity with 
matrix product states and tensor networks, see~\cite{Schollwock11} or~\cite{Verstraete08review}
for  a thorough  review. For the sake of completeness we summarize some basic facts
about matrix product states in \Sect{truncate}. 

\subsection{Construction of the tensor network}
\label{sect:TN}

Let $f\calG$ be one of the cosets $\calC^s_I,\calC^s_X,\calC^s_Y,\calC^s_Z$
defined in \Sect{mld}. Our goal is to compute the coset probability $\pi(f\calG)$.
Let $\pi_1$ be any probability distribution on the single-qubit Pauli
group. For example,  
\[
\pi_1(X)=\pi_1(Y)=\pi_1(Z)=\epsilon/3 \quad \mbox{and} \quad \pi_1(I)=1-\epsilon
\]
for the depolarizing noise with a rate $\epsilon$.
By definition, 
\begin{equation}
\label{eq:coset1}
\pi(f\calG)=\sum_{g\in \calG} \prod_e \pi_1(f_e g_e),
\end{equation}
where the product ranges over all edges of the surface code lattice. 
Let us parameterize $g\in \calG$ by binary variables $\alpha_u,\beta_p\in \{0,1\}$
associated with sites $u$ and plaquettes $p$ such that 
\[
g(\alpha;\beta)= \prod_u (A_u)^{\alpha_u} \cdot \prod_p (B_p)^{\beta_p}.
\]
Here we used a convention $(B_p)^0\equiv I$ and $(A_u)^0\equiv I$. 
Let $e$ be some edge of the surface code lattice with endpoints
$u(e),v(e)$ and adjacent plaquettes $p(e),q(e)$, see \Fig{edge}.
Let $g_e$ be the restriction of $g$ onto the qubit $e$. 
Clearly,  $g_e$ 
depends only on the bits $\alpha_{u(e)},\alpha_{v(e)}$ and $\beta_{p(e)},\beta_{q(e)}$.
Thus we can write
\[
g_e(\alpha;\beta)=g_e(\alpha_{u(e)},\alpha_{v(e)};\beta_{q(e)},\beta_{q(e)}),
\]
where $g_e(i,j;k,l)$ is a function of just four binary variables $i,j,k,l\in \{0,1\}$.
For horizontal edges located at the left or the  right boundary of the lattice the variable
$\alpha_{u(e)}$ or $\alpha_{v(e)}$ respectively is missing.
Likewise, for horizontal edges located at the top or the bottom  boundary
the variable $\beta_{p(e)}$ 
or $\beta_{q(e)}$ respectively is missing. 
We arrive at
\begin{equation}
\label{eq:coset2}
\pi(f\calG)=\sum_\alpha \sum_\beta T(\alpha;\beta),
\end{equation}
where the sums range over binary strings
$\alpha,\beta\in \{0,1\}^{d(d-1)}$ corresponding to all possible configurations
of variables $\alpha_u,\beta_p$ and 
\begin{equation}
\label{eq:coset3}
T(\alpha;\beta)= \prod_e \pi_1(f_eg_e(\alpha_{u(e)},\alpha_{v(e)};\beta_{q(e)},\beta_{q(e)})).
\end{equation}
\begin{figure}[h]
\centerline{\includegraphics[height=2cm]{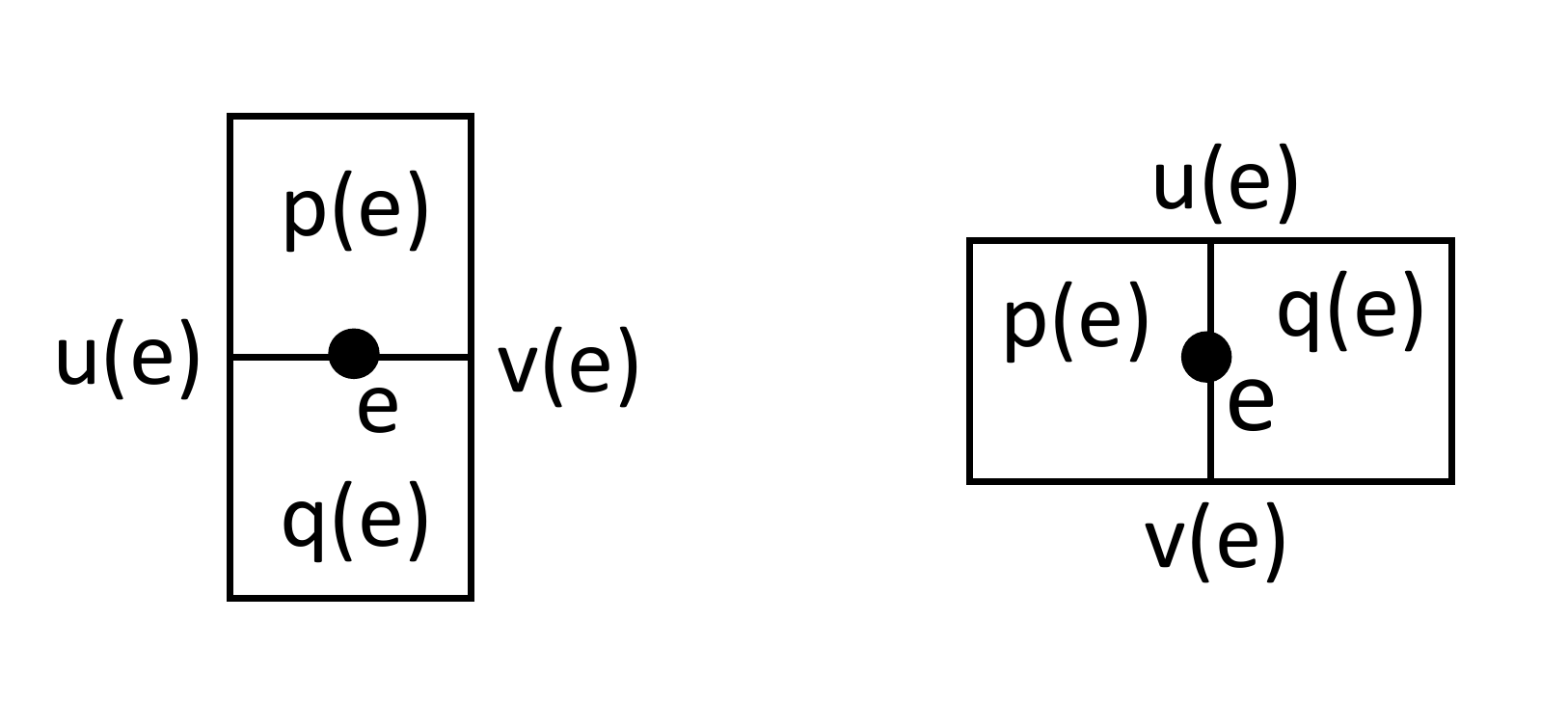}}
\caption{The restriction of a stabilizer $g(\alpha;\beta)$ onto the edge $e$
depends only on the variables $\alpha_{u(e)}$, $\alpha_{v(e)}$
and $\beta_{p(e)}$, $\beta_{q(e)}$.
\label{fig:edge}
}
\end{figure}
The righthand side of \Eq{coset2} coincides with the contraction value of a
properly defined tensor network on a two-dimensional grid. 
To define this tensor network, consider the 
extended surface code lattice shown on \Fig{network}. 
The extended lattice has three types of nodes which we call
$s$-nodes, $h$-nodes, and $v$-nodes. Each
$s$-node represents a location of a stabilizer (either a site 
stabilizer $A_u$ or plaquette stabilizer $B_p$)
while $h$-nodes and $v$-nodes represent code qubits  located on 
horizontal and vertical edges of the original surface code lattice respectively.
We shall refer to edges of the extended lattice as {\em links} to 
distinguish them from edges of the original surface code lattice. 

\begin{figure}[h]
\centerline{\includegraphics[height=4cm]{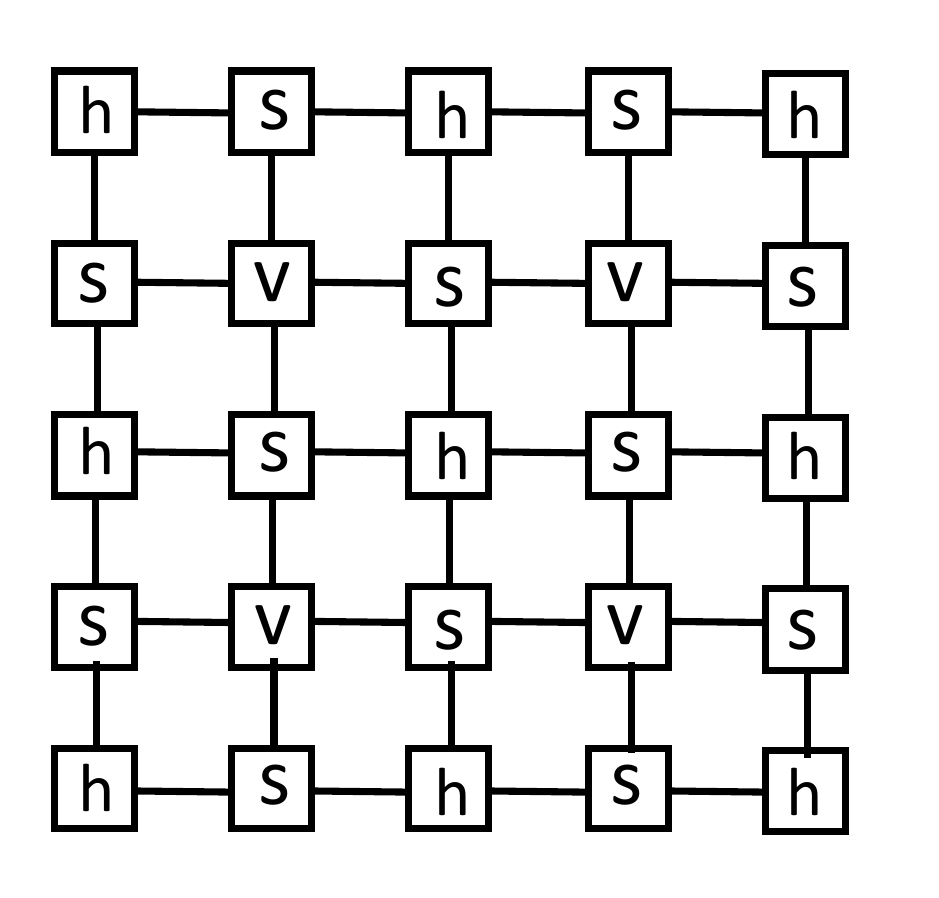}}
\caption{The extended surface code lattice for $d=3$.  
Locations of stabilizers are represented by $s$-nodes. 
Code qubits located on horizontal and vertical edges of the original lattice
are represented by $h$-nodes and $v$-nodes respectively. 
In general, the extended lattice has dimensions $(2d-1)\times (2d-1)$. 
\label{fig:network}
}
\end{figure}

Consider any configuration of variables $\alpha,\beta$ and 
the corresponding term $T(\alpha;\beta)$ in \Eq{coset2}.
For each site stabilizer $A_u$ let us copy the corresponding variable
$\alpha_u$ to all links incident to the $s$-node $u$.
Likewise, for each plaquette stabilizer $B_p$ let us copy the corresponding variable
$\beta_p$ to all links incident to the $s$-node $p$.
We obtain a labeling of the links  by binary variables $\gamma(\alpha;\beta)$
with the property that all links incident to any $s$-node have the same label. 
Let us call such a link labeling {\em valid}. 
By definition, $T(\alpha;\beta)$ is a product of terms 
\[
T_e(\alpha;\beta)\equiv  \pi_1(f_eg_e(\alpha_{u(e)},\alpha_{v(e)};\beta_{q(e)},\beta_{q(e)}))
\]
associated with $h$-nodes and $v$-nodes $e$ of the extended lattice. 
Since $\alpha$ and $\beta$ are uniquely determined by the link labeling $\gamma(\alpha;\beta)$,
we can also write $T_e(\alpha;\beta)$ as a function of $\gamma$, that is,
$T_e(\alpha;\beta)=T_e(\gamma)$.  This shows that 
\begin{equation}
\label{eq:coset3}
\pi(f\calG)=\sum_{\mathrm{valid} \; \gamma} \; \prod_{e\in h,v}\;  T_e(\gamma),
\end{equation}
where the product is over all $h$-nodes and $v$-nodes and 
 the sum ranges over all valid link labelings.
We can now extend
the sum in \Eq{coset3} to {\em all} link labelings  $\gamma$ by adding extra 
terms $T_e(\gamma)\in \{0,1\}$ associated with $s$-nodes $e$ such that
$T_e(\gamma)=1$ iff all links incident to $e$ have the same label and $T_e(\gamma)=0$
otherwise. We arrive at 
\begin{equation}
\label{eq:coset4}
\pi(f\calG)=\sum_{\gamma} \; \prod_{e}\;  T_e(\gamma),
\end{equation}
Now the product ranges over all nodes of the extended lattice and the
sum ranges over all link labelings. Furthermore, 
by construction,  each term $T_e(\gamma)$ depends only
on the labels of links incident to the node $e$.  
The expression in the righthand side of \Eq{coset4} is known as a
contraction value of the tensor network defined by the collection of tensors $T_e(\gamma)$.
Tensor networks are usually represented by diagrams like the one shown  on \Fig{network}
such that each box on the diagram carries a tensor with several indexes. 
Indexes of a tensor are associated with the links emanating from the corresponding box. 
Diagrams representing the tensors $T_e(\gamma)$ are shown on
Eqs.~(\ref{eq:snode},\ref{eq:hnode},\ref{eq:vnode}). All tensor
indexes  $i,j,k,l$ on these diagrams take values $0,1$.
For tensors located at the boundary some of the indexes may be missing. 
Note that the order of arguments of $g_e$ is interchanged in Eqs.~(\ref{eq:hnode},\ref{eq:vnode}). 
This is simply because the qubits located on horizontal edges ($h$-nodes)
have site stabilizers on the left and on the right whereas qubits located on vertical edges
($v$-nodes) have site stabilizers on the top and on the bottom. 

\begin{equation}
\label{eq:snode}
\mbox{$s$-node:} \quad
\raisebox{-0.95cm}{\includegraphics[height=2cm]{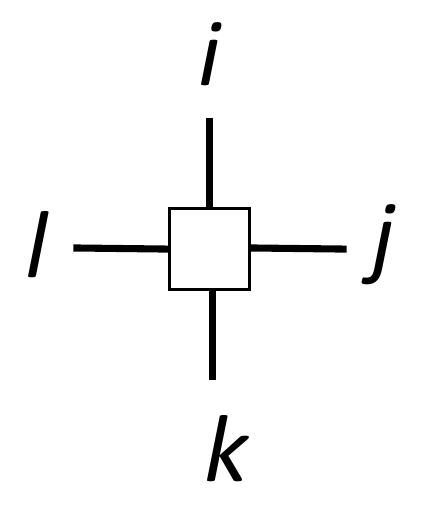}}
\quad = \left\{ \ba{rcl}
1 &\mbox{if} & i=j=k=l\\
0 && \mbox{otherwise} \\
\ea \right.  
\end{equation}
\begin{equation}
\label{eq:hnode} 
\mbox{$h$-node:}\quad 
\raisebox{-0.95cm}{\includegraphics[height=2cm]{snode.pdf}}= \pi_1(f_e g_e(j,l;i,k)) 
\end{equation}
\begin{equation}
\label{eq:vnode} 
\mbox{$v$-node:} \quad 
\raisebox{-0.95cm}{\includegraphics[height=2cm]{snode.pdf}}= \pi_1(f_e g_e(i,k; j,l)) 
\end{equation}

\subsection{Approximate contraction algorithm} 

Let $\mps{\chi}$
and $\mpo{\chi}$ be the set of matrix product states and
matrix product  operators defined on a chain of  $2d-1$ qubits
and having  the bond dimension $\chi$.  In this section we shall 
identify a matrix product state (operator) with the corresponding tensor network. 
Consider a partition of the extended surface code lattice into columns shown on \Fig{slices}.
Each column $V^j$ and each internal column $H^j$ 
defines a matrix product operator $\hat{V}^j\in \mpo{2}$ and $\hat{H}^j\in \mpo{2}$
respectively. 
The first and the last columns $H^1,H^d$ define matrix product states
$\hat{H}^1,\hat{H}^d\in \mps{2}$. Here we identify horizontal links of the lattice
with physical indexes of MPO and MPS, while vertical links correspond to virtual indexes.
By definition, contracting a consecutive pair of columns  is equivalent to taking the product of 
the corresponding MPOs. Thus \Eq{coset4} can be rewritten as 
\begin{equation}
\label{eq:network2}
\pi(f\calG)=\langle \hat{H}^d|\hat{V}^{d-1}  \cdots \hat{H}^2 \hat{V}^1 |\hat{H}^1\rangle.
\end{equation}

\begin{figure}[h]
\centerline{\includegraphics[height=4.5cm]{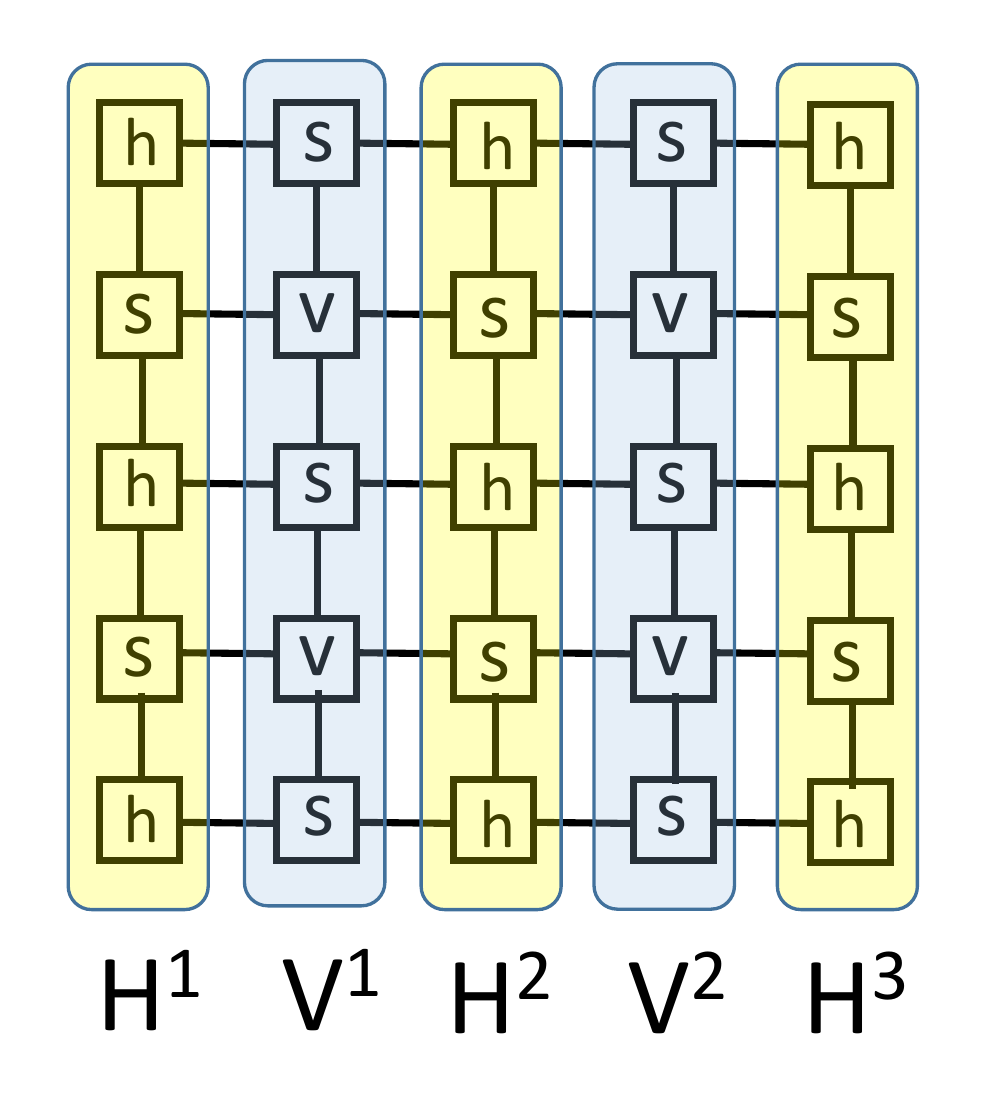}}
\caption{Partition of the extended lattice into  `horizontal' columns $H^1,\ldots,H^d$
and `vertical' columns $V^1,\ldots,V^{d-1}$ (here $d=3$).
\label{fig:slices}
}
\end{figure}
To approximate the righthand side of \Eq{network2}
we shall employ the  algorithm
proposed by Murg, Verstraete and Cirac~\cite{Verstraete04,Murg07}.
The approximation accuracy of the algorithm is controlled by 
an integer parameter $\chi\ge 2$ such that the algorithm becomes exact if
$\chi$ is exponentially large in $d$. 
At each step of the algorithm we maintain a state $\psi\in \mps{\chi}$.
Such a state can be described by a list of $2d-1$ tensors
of dimension $2\times \chi\times \chi$ 
 which requires 
$O(d\chi^2)$ real parameters. We begin by initializing $\psi=\hat{H}^1$.
Note that $\hat{H}^1\in \mps{2}\subseteq \mps{\chi}$.
Each step of the algorithm updates $\psi$ according to
$\psi\to \hat{H}^j \psi$ (even steps) or $\psi\to \hat{V}^j \psi$ (odd steps). 
This update is realized simply by taking the product of tensors
of $\psi$ with the respective tensors of $\hat{H}^j$ or $\hat{V}^j$ which takes
time $O(d\chi^2)$. 
Since $\hat{H}^j$ and $\hat{V}^j$ map $\mps{\chi}$ to $\mps{2\chi}$,
extra measures have to be taken to reduce the bond dimension after each update. 
To this end we apply the truncation algorithm described in Section~4.5 of Ref.~\cite{Schollwock11}.
We shall use a function  Truncate$()$ that
takes as input a state $\phi\in \mps{2\chi}$ and returns a state $\psi\in \mps{\chi}$
approximating $\phi$. Such an approximation is obtained by computing the Schmidt
decomposition of $\phi$ across each bipartite cut of the chain and 
retaining only the $\chi$ largest Schmidt coefficients. 
A detailed implementation of the function Truncate$()$ is described in the next section. 
The  last step of the algorithm is to compute the inner product between the final
state $\psi\in \mps{\chi}$ and $\hat{H}^d\in \mps{2}$. This can be done
in time $O(d\chi^3)$ by applying the standard contraction method for MPS.
As we explain in the next section, each call to the function Truncate$()$
involves $2d-1$ QR-decompositions and SVD-decompositions
on matrices of size $2\chi\times 2\chi$ and $2\chi\times \chi$ respectively, 
which takes time $O(d\chi^3)$.
Since we need one truncation for each column of the lattice, the overall
running time of the algorithm is $O(d^2\chi^3)=O(n\chi^3)$. 
The above steps can be summarized as follows. 

\begin{center}
\fbox{\parbox{0.9\linewidth}{
\begin{algorithmic}
\State{\hspace{2.2cm} {\bf Algorithm~2}}
\State{{\bf Input:} Pauli operator $f$}
\State{{\bf Output:} Approximation to $\pi(f\calG)$}
\State{}
\State{$\psi\gets \hat{H}^1$}
\For{$j=1$ to $d-2$}
\State{$\psi\gets$\Call{Truncate}{$\hat{V}^j \psi$}}
\State{$\psi\gets$\Call{Truncate}{$\hat{H}^{j+1} \psi$}}
\EndFor
\State{$\psi\gets$\Call{Truncate}{$\hat{V}^{d-1} \psi$}}
\State{\Return $\langle \hat{H}^d |\psi\rangle$}
\end{algorithmic}
}}
\end{center}
\[
\quad
\]

\subsection{Truncation of a matrix product state}
\label{sect:truncate}

In this section we describe  implementation of the function Truncate$()$ in Algorithm~2. 
Our implementation closely follows Section~4.5 of Ref.~\cite{Schollwock11}. 
For the sake of completeness, we begin by summarizing the necessary facts about matrix product states.
Below we use a notation $L\equiv 2d-1$ for the number of qubits per column of the lattice. 

A matrix product state $|\psi\rangle\in (\CC^2)^{\otimes L}$ 
describing a chain of $L$ qubits is 
defined by a list of $2L$  matrices $A_0(s),A_1(s)$, where
$s=1,\ldots,L$ is a qubit index (site of the chain).
 Any amplitude of $\psi$  in the standard basis 
is expressed as a product  of $L$ matrices
\begin{equation}
\label{eq:MPS}
\langle x|\psi\rangle= A_{x_1}(1)A_{x_2}(2) \cdots A_{x_L}(L) , \quad x\in \{0,1\}^L.
\end{equation}
We shall use a shorthand notation $A(s)$ for the pair of matrices $A_0(s), A_1(s)$
at some particular qubit $s$. 
Likewise $A$ will stand for the full matrix product state. 
The $L$-qubit state defined in \Eq{MPS} will be denoted $\psi(A)$. 
Let $r(s)$ and $c(s)$ be the number of rows and columns respectively in $A_{0,1}(s)$ (we shall always assume $A_0(s)$ and $A_1(s)$ 
have the same dimensions). 
Since we want the 
the product of matrices in \Eq{MPS} to be a $1\times 1$ matrix (a complex number),
dimensions of the matrices must satisfy
\[
r(1)=1, \quad c(L)=1, \quad c(s)=r(s+1) \quad \mbox{for $1\le s<L$}.
\]
A  matrix product state is said to have a {\em bond dimension} $\chi$ 
iff $r(s)\le \chi$ and $c(s)\le \chi$ for all qubits $s$. 
Let $\mps{\chi}$ be the set of all matrix product states $A$ 
on $L$ qubits  with the bond dimension $\chi$.
We shall say that $A(s)$ has a  {\em left canonical form} (LCF)
or {\em right canonical form} (RCF)
iff 
\begin{equation}
\label{eq:LCF}
A_0(s)^\dag A_0(s) + A_1(s)^\dag A_1(s) =I_{c(s)}
\end{equation}
or 
\begin{equation}
\label{eq:RCF}
A_0(s) A_0(s)^\dag + A_1(s) A_1(s)^\dag =I_{r(s)}
\end{equation}
respectively. Here $I_n$ denotes the identity matrix of size $n\times n$.
The importance of LCF and RCF comes from the following lemma.
Here and below we use a notation $e^i$ for the column vector 
$[0,\ldots,0,1,0,\ldots,0]\tp$ with `$1$' at the $i$-th coordinate.
\begin{lemma}
\label{lem:LCF}
Suppose $A(s)$ has LCF for $s=1,\ldots,m$.
For each $\alpha=1,\ldots,c(m)$ define a state $\phi_\alpha\in (\CC^2)^{\otimes m}$
with amplitudes 
\begin{equation}
\label{eq:phi}
\langle x|\phi_\alpha\rangle = A_{x_1}(1) A_{x_2}(2) \cdots A_{x_m}(m) e^\alpha,
\quad x\in \{0,1\}^m.
\end{equation}
Then $\phi_\alpha$ form an orthonormal family of vectors, i.e.,
$\langle \phi^\beta|\phi^\alpha\rangle=\delta_{\alpha,\beta}$ 
for all $1\le \alpha,\beta\le c(m)$.
\end{lemma}
\begin{proof}
Indeed, using the definition of $\phi^\alpha$ the inner product $\langle \phi^\beta|\phi^\alpha\rangle=
\sum_x \langle\phi^\beta|x\rangle\langle x|\phi^\alpha\rangle$ can be written as 
\[
\sum_{x} (e^\beta)\tp A_{x_m}(m)^\dag \cdots A_{x_1}(1)^\dag A_{x_1}(1) \cdots A_{x_m}(m) e^\alpha,
\]
where the sum runs over $x\in \{0,1\}^m$.
The LCF at qubit $1$ implies $\sum_{x_1} A_{x_1}(1)^\dag A_{x_1}(1)=I_{c(1)}$.
Hence $\langle \phi^\beta|\phi^\alpha\rangle$ is equal to
\[
\sum_{x} (e^\beta)\tp A_{x_m}(m)^\dag \cdots A_{x_2}(2)^\dag A_{x_2}(2) \cdots A_{x_m}(m) e^\alpha,
\]
where the sum runs over $x\in \{0,1\}^{m-1}$. Applying the same argument to the remaining qubits
one arrives at $\langle \phi^\beta|\phi^\alpha\rangle=(e^\beta)\tp  e^\alpha=\delta_{\alpha,\beta}$.
\end{proof}
Exactly the same arguments show that if $A(s)$ has RCF for all $s>m$
then  states $\theta^\alpha\in (\CC^2)^{\otimes (L-m)}$ with amplitudes
\begin{equation} 
 \label{eq:theta}
\langle y|\theta^\beta\rangle = (e^\beta)\tp A_{y_1}(m+1) \cdots A_{y_{L-m}}(L)
\end{equation}
form an orthonormal family for $1\le \beta\le r(m+1)$.

The first step of the function Truncate is transforming all matrices $A(s)$ to LCF.
We shall describe this step by a function LeftCanonical($A$) that takes as input 
a matrix product state 
$A\in \mps{\chi}$ and  returns a 
pair $(\Gamma,B)$, where $\Gamma\in \CC$ is a scalar and $B\in \mps{\chi}$ 
is a matrix product state such that $\psi(A)=\Gamma\cdot \psi(B)$
and $B$ has LCF at every qubit.  We shall define
LeftCanonical($A$)  by the following algorithm.

\begin{center}
\fbox{\parbox{0.9\linewidth}{
\begin{algorithmic}
\Function{$(\Gamma,B)$=LeftCanonical}{$A$} 
\For{$s=1$ to $L$}
\State{$(Q,R)\gets \mbox{QR-decomposition of $A(s)$} $}
\State{\hspace{1.4cm} as defined in Eqs~(\ref{eq:QR},\ref{eq:QR0})}
\State{$B_0(s) \gets Q_0$}
\State{$B_1(s) \gets Q_1$}
 \If {$s<L$}
\State{$A_0(s+1) \gets RA_0(s+1)$}
\State{$A_1(s+1) \gets RA_1(s+1)$}
 \Else 
 \State{$\Gamma \gets R$}
 \EndIf
 \EndFor
 \EndFunction
\end{algorithmic}
}}
\end{center}

Let us explain  the QR-decomposition step in the above algorithm
and prove its correctness.
Consider any qubit $s$ and represent $A(s)$ as a block matrix 
\begin{equation}
\label{eq:QR}
A(s)=\left[ \ba{c} A_0(s) \\ A_1(s) \\ \ea\right].
\end{equation}
Note that $A(s)$ has $2r(s)$ rows and $c(s)$ columns. 
Let $m=\min{\{ c(s),2r(s)\}}$. Applying the `economic'  QR-decomposition to $A(s)$ one gets
\begin{equation}
\label{eq:QR0}
A(s)=QR,
\end{equation}
where $Q$ has dimensions $2r(s)\times m$,   $R$ has dimensions $m\times c(s)$,
and  columns of $Q$ form an orthonormal family of vectors, that is,
$Q^\dag Q=I_m$. Finally, $R$ is an upper triangular matrix (this property will not be important for us). 
Let us write 
\[
Q=\left[ \ba{c} Q_0 \\ Q_1 \\ \ea\right],
\]
where $Q_{0,1}$ have dimensions $r(s)\times m$. The property $Q^\dag Q=I_m$
is equivalent to $Q_0^\dag Q_0 + Q_1^\dag Q_1=I_m$. 
Hence $B(s)$ defined in the above algorithm has LCF. 
Note that dimensions of $B_{0,1}(s)$ may or may not be equal to the ones of $A_{0,1}(s)$. 
Let $A'_{0,1}(s+1)=RA_{0,1}(s+1)$ be the updated version of $A(s+1)$
defined in the algorithm. Obviously $A_x(s)A_y(s+1)=B_x(s)A_y'(s+1)$ for any
$x,y=0,1$. Thus $A_{x_1}(1)\cdots A_{x_L}(L)$ is equal to 
\[
B_{x_1}(1) \cdots B_{x_s}(s) A_{x_{s+1}}'(s+1) A_{x_{s+2}}(s+2) \cdots A_{x_L}(L)
\]
for all $x\in \{0,1\}^L$ and for all $s=1,\ldots,L-1$. The last step of the algorithm
($s=L$) applies a QR-decomposition to a column vector $A(L)$, possibly updated by the 
previous step of the algorithm. Hence $Q$ is  a unit-norm
column vector of size $2r(L)$ while $R$ is a scalar which determines normalization
of the overall state. This proves that $\psi(A)=\Gamma\cdot \psi(B)$ and $B$ has LCF
at every qubit. 

Suppose  the input matrix product state $A$ has bond dimension $\chi$.
Then the computational cost of each QR-decomposition is $O(\chi^3)$. Therefore, the 
function LeftCanonical($A$) can be computed in time $O(L\chi^3)$. Since no step
of the algorithm increases dimensions of the matrices, the final matrix product state
$B$ also has bond dimension $\chi$.

We are now ready to describe the function Truncate. 
Choose any integer $1\le m\le L$ and partition the chain as $\calL\cup m\cup \calR$, where 
\[
\calL=\{1,\ldots,m-1\} \quad \mbox{and} \quad \calR=\{m+1,\ldots,L\}.
\]
Consider a matrix product state $A$ such that $A(s)$ has LCF for all
$s\in \calL$ and RCF for all $s\in \calR$. Suppose also that 
the matrices $A_{0,1}(s)$ have dimensions at most $\chi$
for all $s\in \calR$ and at most $\tilde{\chi}$ for all $s\in \calL$. We assume that
$\tilde{\chi}>\chi$ (we shall be interested in the case $\tilde{\chi}=2\chi$). 
Using the orthonormal families of states $\phi^\alpha$ and $\theta^\alpha$ defined in Eqs.~(\ref{eq:phi},\ref{eq:theta})
one can write $\psi(A)$ as
\begin{equation}
\label{eq}
\psi(A)=\sum_{\alpha=1}^{r(m)} \sum_{\beta=1}^{c(m)} \sum_{x=0,1} A_x(m)_{\alpha,\beta} 
|\phi^\alpha\otimes x \otimes \theta^\beta\rangle. 
\end{equation}
We shall compute the Schmidt decomposition of $\psi(A)$ with respect to the
partition $\calL \cup \{m,\calR\}$ and truncate this decomposition by retaining only the
$\chi$ largest Schmidt coefficients. To this end consider the singular value decomposition (SVD)
of $A(m)$, namely, 
\begin{equation}
\label{eq:svd}
A(m)\equiv \left[ \ba{c|c} A_0(m) &  A_1(m) \\ \ea\right]=
U S V^\dag,
\end{equation}
where the matrices $U,S,V$ have dimensions 
\[
\dim{U}=r(m)\times n, \quad \dim{S}=n\times n, \quad \dim{V}=2c(m)\times n,
\]
with  
\[
n=\min{\{ r(m),2c(m)\}}.
\]
The matrix $S$ is diagonal such that $S_{i,i}$ is the
$i$-th largest singular value of $A(m)$. The matrices $U$ and $V$ are isometries, that is,
\[
U^\dag U = V^\dag V = I_n.
\]
Let us represent $V$ as a block matrix 
\begin{equation}
\label{eq:Vblock}
V=\left[ \ba{c} V_0 \\ V_1 \\ \ea\right],
\end{equation}
where $V_0$ and $V_1$ have dimensions $c(m)\times n$. 
Using the above SVD one can rewrite $\psi(A)$ as
\begin{equation}
\label{eq:svd1}
\psi(A)=\sum_{i=1}^n S_{i,i} |\hat{\phi}^i\rangle \otimes |\hat{\theta}^i\rangle,
\end{equation}
where $\hat{\phi}^i$ and $\hat{\theta}^i$  are orthonormal family of $n$ states
defined as
\begin{equation}
\label{eq:svd2}
|\hat{\phi}^i\rangle = \sum_{\alpha=1}^{r(m)} U_{\alpha,i} \, |\phi^\alpha\rangle
\end{equation}
and
\begin{equation}
\label{eq:svd3}
|\hat{\theta}^i\rangle = \sum_{\beta=1}^{c(m)} (V_0^*)_{\beta,i} |0\otimes \theta^\beta\rangle
+ (V_1^*)_{\beta,i} |1\otimes \theta^\beta\rangle.
\end{equation}
We conclude that \Eq{svd1} defines the Schmidt decomposition of $\psi(A)$ with respect
to the partition $\calL\cup \{m,\calR\}$, while $S_{i,i}$ are the Schmidt coefficients. 
The best rank-$\chi$ approximation to $\psi(A)$ which we denote $\psi'(A)$ is obtained from \Eq{svd1} by retaining
$\chi$ largest Schmidt coefficients, that is, 
\begin{equation}
\label{eq:svd4}
\psi'(A)=\sum_{i=1}^\chi S_{i,i} |\hat{\phi}^i\rangle \otimes |\hat{\theta}^i\rangle.
\end{equation}
Decompose  matrices $U,S,V$  into blocks such that 
\begin{equation}
\label{eq:trun}
U=\left[ \ba{c|c} U' &  U'' \\ \ea \right], \; 
S=\left[ \ba{cc} S' & 0 \\ 0 & S'' \\ \ea \right], \; 
V'=\left[ \ba{c|c} V' & V'' \\ \ea \right].
\end{equation}
By definition, $U',S',V'$ have dimensions
\[
\dim{U'}=r(m)\times \chi, \; 
\dim{S'}=\chi \times \chi, \;  \dim{V'}=2c(m)\times \chi.
\]
Furthermore, $S'$ is a square diagonal matrix that contains $\chi$ largest singular values
of $A(m)$, while  $U'$ and $V'$ are isometries, that is, 
$(U')^\dag U'=I_\chi$ and 
$(V')^\dag V'=I_\chi$.  
We conclude that   $\psi'(A)=\psi(A')$, where $A'(s)=A(s)$ for $s\in \calR$
and for $s\in \calL\setminus m$, 
\[
A'_{0,1}(m-1)=A_{0,1}(m-1)U'S' \quad \mbox{and} \quad A'(m)=(V')^\dag.
\]
The fact that $V'$ is an isometry implies that 
$A'(m)$ has RCF, so we can apply the above procedure
again with $\calL=\calL\setminus \{m-1\}$ and $\calR=\calR\cup \{m\}$. 
Starting from $m=L$ and moving towards the left boundary of the chain
one can reduce the bond dimension from $\tilde{\chi}$ to $\chi$. 
The above truncation algorithm can be summarized as follows.

\begin{center}
\fbox{\parbox{0.9\linewidth}{
\begin{algorithmic}
\Function{Truncate}{$A$} 
\State{\Call{$(\Gamma,A)\gets$LeftCanonical}{$A$}}
\For{$m=L$ to $1$}
\State{$(U,S,V)\gets \mbox{svd-decomposition of $A(m)$} $}
\State{\hspace{1.7cm} defined in \Eq{svd}}
\State{$U',S',V' \gets \mbox{submatrices of $U,S,V$}$}
\State{\hspace{1.8cm}  defined in \Eq{trun}}
\State{$A_{0,1}(m-1) \gets A_{0,1}(m-1)U'S'$}
\State{$A_{0,1}(m) \gets (V'_{0,1})^\dag$} 
\EndFor
\State{\Return $\Gamma\cdot A$}
 \EndFunction
\end{algorithmic}
}}
\end{center}
Here we decomposed $V'$ into blocks $V'_0$ and $V'_1$ 
similar to \Eq{Vblock}.

\section{Numerical results}
\label{sect:numerics}

We have studied the following combinations of noise models and
decoders:
\begin{enumerate}
\item $X$-noise, ML decoder
\item $X$-noise, MPS decoder
\item $X$-noise, MWM decoder
\item Depolarizing noise, MPS decoder
\item Depolarizing noise, MWM  decoder
\end{enumerate}
For each of the above combinations we  
estimated the probability of a logical error --- the decoding outcome
in which the recovery operator 
differs from the actual error  by a logical Pauli operator
(we do not differentiate between $\overline{X},\overline{Y}$, or $\overline{Z}$ logical errors). 
The performance of each decoder was measured in terms of its error threshold and its
{\em badness} parameter --- the ratio between the logical error probabilities of a given decoder
and the best available decoder for the considered noise model. Thus badness $\ge 1$
for any decoder with smaller values indicating better decoders. 
 The exact ML decoder and MPS decoders were implemented as described in 
\Sect{Xnoise} and \Sect{XYZnoise} respectively. The MWM decoder was implemented
by a reduction from the minimum weight perfect matching problem to the maximum weight
matching problem as described in Ref.~\cite{BV13}.

\begin{figure}[h]
\centerline{\includegraphics[height=6.5cm]{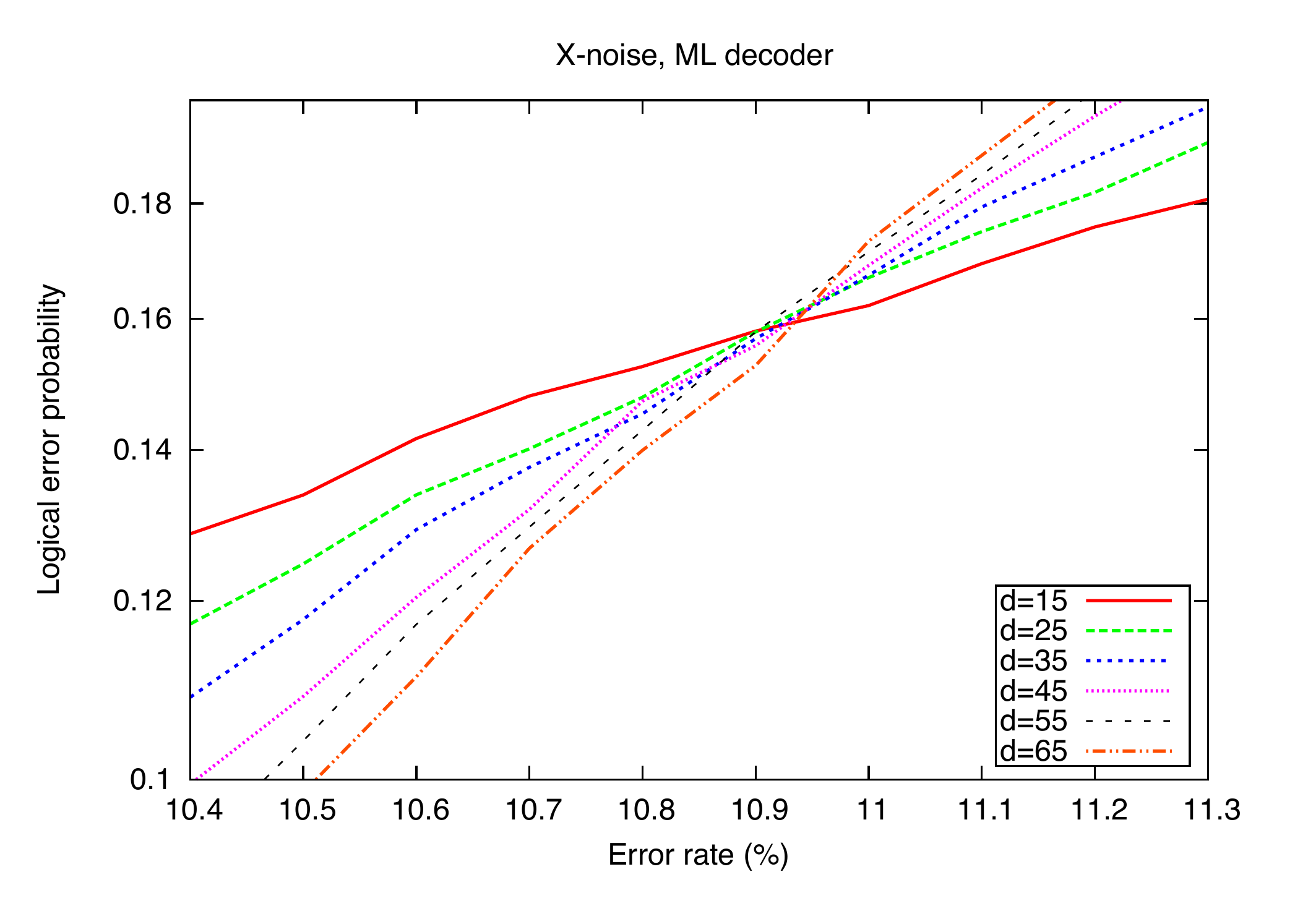}}
\caption{X-noise: exact implementation of the ML decoder.
The data suggest that the threshold error rate
$\epsilon_0$ is between $10.9\%$ and $11\%$, which is in a good agreement with the estimate
$\epsilon_0=10.93(2)$
of Ref.~\cite{Merz01} which calculated the phase transition point in the respective spin model.
Each curve has data points at error
rates $\epsilon=10.4, 10.5,\ldots,11.3\%$. To compute the logical error
probability, at least $5,000$ failed error correction trials have been accumulated
for each datapoint. 
\label{fig:Xth}
}
\end{figure}

Let us first discuss our results for the X-noise.  
The threshold error rate
$\epsilon_0$ of the ML decoder coincides with the critical density of anti-ferromagnetic bonds in the
random-bond Ising model  on the Nishimori line~\cite{DKLP01}. The latter has been estimated
numerically by Mertz and Chalker~\cite{Merz01} who found $\epsilon_0=10.93(2)\%$.
Our data shown at \Fig{Xth} suggest that $10.9\%\le \epsilon_0\le 11\%$ which is 
in a good agreement with the estimate of Ref.~\cite{Merz01}. 
For comparison, the MWM decoder is known to have the threshold $\epsilon_0\approx 10.31\%$,
see~\cite{Wang03}.

\onecolumngrid
\begin{center}
\begin{figure}[h]
\centerline{\includegraphics[height=8cm]{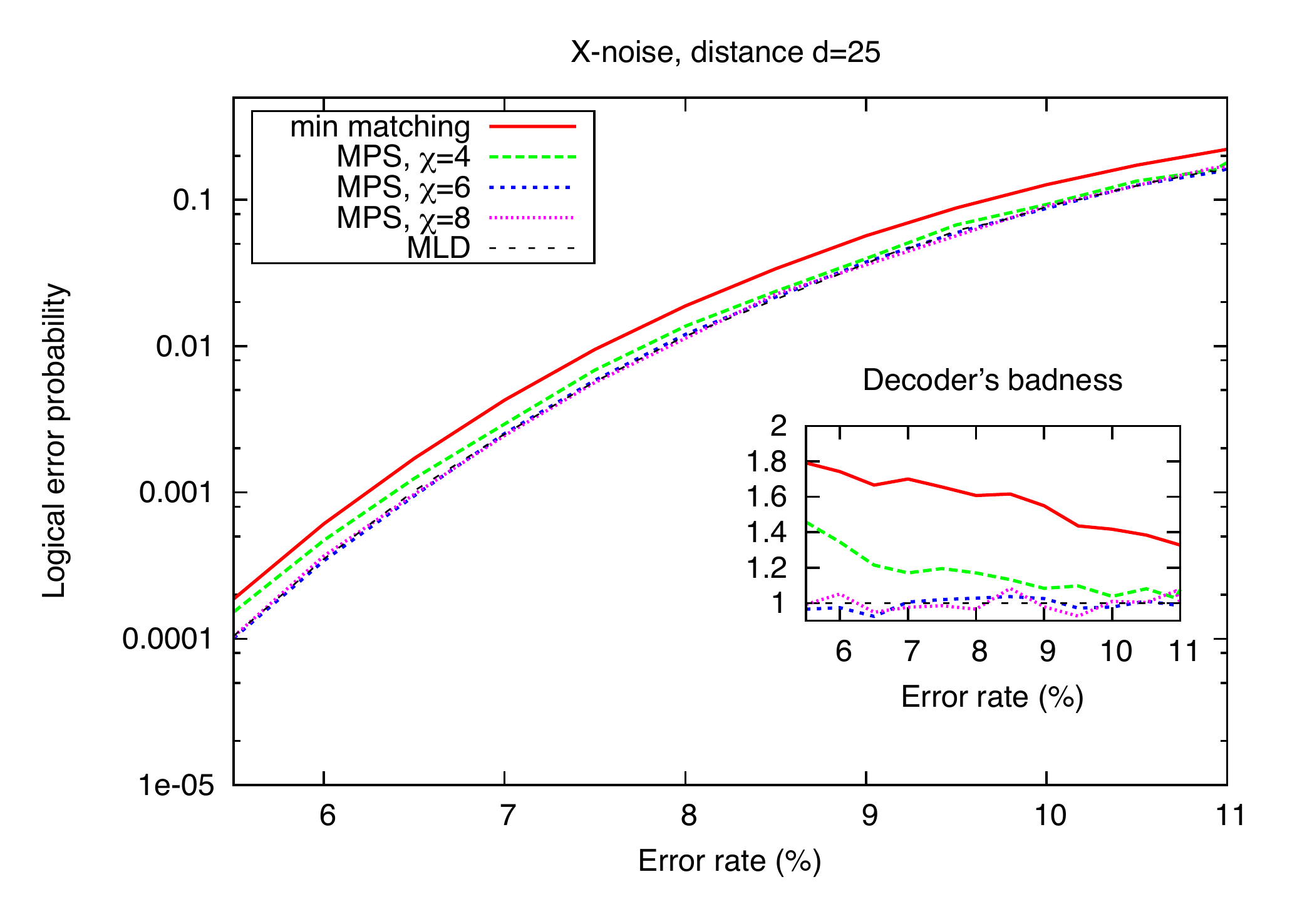}}
\caption{X-noise: exact and approximate implementations of the ML decoder.
Logical error probability as a function of the error rate $\epsilon$
is shown.  The curves representing the exact MLD and MPS decoders with
$\chi=6,8$ are too close to be distinguishable on the main plot. 
The red curve represents the standard minimum weight matching decoder. 
The inset shows 'badness' of various decoders as a function of the error rate.
We define the badness 
as the ratio between  logical error probabilities of a given decoder and the optimal decoder (MLD).
Each curve has data points at error
rates $\epsilon=5, 5.5,6,\ldots,11\%$. To compute the logical error
probability, at least $1,000$ failed error correction trials have been accumulated
for each datapoint. 
\label{fig:Xcomp}
}
\end{figure}
\end{center}
\twocolumngrid

The performance of  different decoders  for a fixed code distance $d=25$
and a wide range of error rates is shown at \Fig{Xcomp}. 
We observed that the MWM decoder remains nearly optimal 
for all simulated  error rates with the badness parameter $\le 2$,
 even though for these error rates the logical error probability changes by several orders of magnitude. 
The slight difference between MLD and the MWM decoder can be explained
by the fact that the latter ignores the error degeneracy~\cite{Stace10}. The data shown on \Fig{Xcomp}
suggests that for X-noise ignoring the error degeneracy does not have a significant impact on the 
performance, even for large error rates and large code distances.

Perhaps more surprisingly, \Fig{Xcomp} demonstrates 
 that the MPS decoder with a relatively small bond dimension
$\chi=6,8$ is virtually indistinguishable from the optimal one in terms of the logical error probability.
This serves as a numerical proof of correctness for the MPS decoder.  

We observed numerically that the exact MLD algorithm described in \Sect{Xnoise} becomes
very sensitive to rounding errors in the regime of large code distances and small error rates.
One way to suppress rounding errors is to enforce an orthogonality condition
$M^TM=I$ on the covariance matrix $M$ in Algorithm~1.
The  orthogonality condition  is satisfied automatically if all arithmetic operations are perfect
(because $M$ represents a covariance matrix of a pure Gaussian state, see \Sect{gaussian} for details). 
In practice, we observed that the orthogonality can be quickly lost if no special measures are taken. 
A simple
and computationally cheap solution of the above problem is to
compute the QR-decomposition $M=QR$, where $Q$ is an orthogonal matrix
and $R$ is an upper-triangular matrix. Note that $M^TM=I$ is possible only if 
$R$ is a diagonal matrix with entries $\pm 1$ on the diagonal. 
This form of $R$ can be easily enforced by setting all off-diagonal entries of $R$ to zero
and replacing each diagonal entry $R_{i,i}$ by the sign of $R_{i,i}$.  
Let $\tilde{R}$ be the resulting diagonal matrix. 
We found that replacing $M$ by $M'\equiv (Q\tilde{R}-(Q\tilde{R})^T)/2$  after
each call to the functions SimulateHorizontal and SimulateVertical in Algorithm~1
makes the algorithm more stable against rounding errors. 

Let us now discuss the  depolarizing noise. 
In this case we only have an approximate implementation of MLD 
with no direct means of estimating  the approximation precision.
Hence the first natural question is  whether the MPS decoder
with a fixed bond dimension $\chi$ has a non-zero error threshold $\epsilon_0$. 
Our data suggests (although not conclusively) that the answer is `yes'.
Most importantly, we observed an exponential decay of the logical error probability
as a function of the code distance $d$ for a fixed  error rate, see \Fig{XYZdecay},
where we used $\chi=6$. Assuming that the observed decay does not saturate
for larger $d$, the data shown at \Fig{XYZdecay} gives a lower bound $\epsilon_0 \ge 14\%$.
The logical error probability as a function of the error rate for a fixed $d$
is shown on \Fig{XYZth} which also exhibits a  typical threshold-like behavior
 and suggests that $17\%\le \epsilon_0\le 18.5\%$.
Previously studied approximate versions of MLD such as the renormalization group decoder~\cite{Duclos10}
and the Markov chain decoder~\cite{Hutter14},
as well the MWM decoder~\cite{Wang09threshold} have error thresholds between  $15\%$ and $16\%$.
The threshold of the exact ML decoder corresponding to the phase transition point in the 
disordered eight-vertex Ising model is known to be $\epsilon_0\approx 18.9(3)\%$, see  Ref.~\cite{BombinPRX12}.
Since the correlation length of the Ising model diverges at the phase transition point, 
we expect that the MPS decoder can only achieve this optimal threshold if the bond dimension
$\chi$ is a growing function of the code distance $d$. 

\begin{figure}[h]
\centerline{\includegraphics[height=6.5cm]{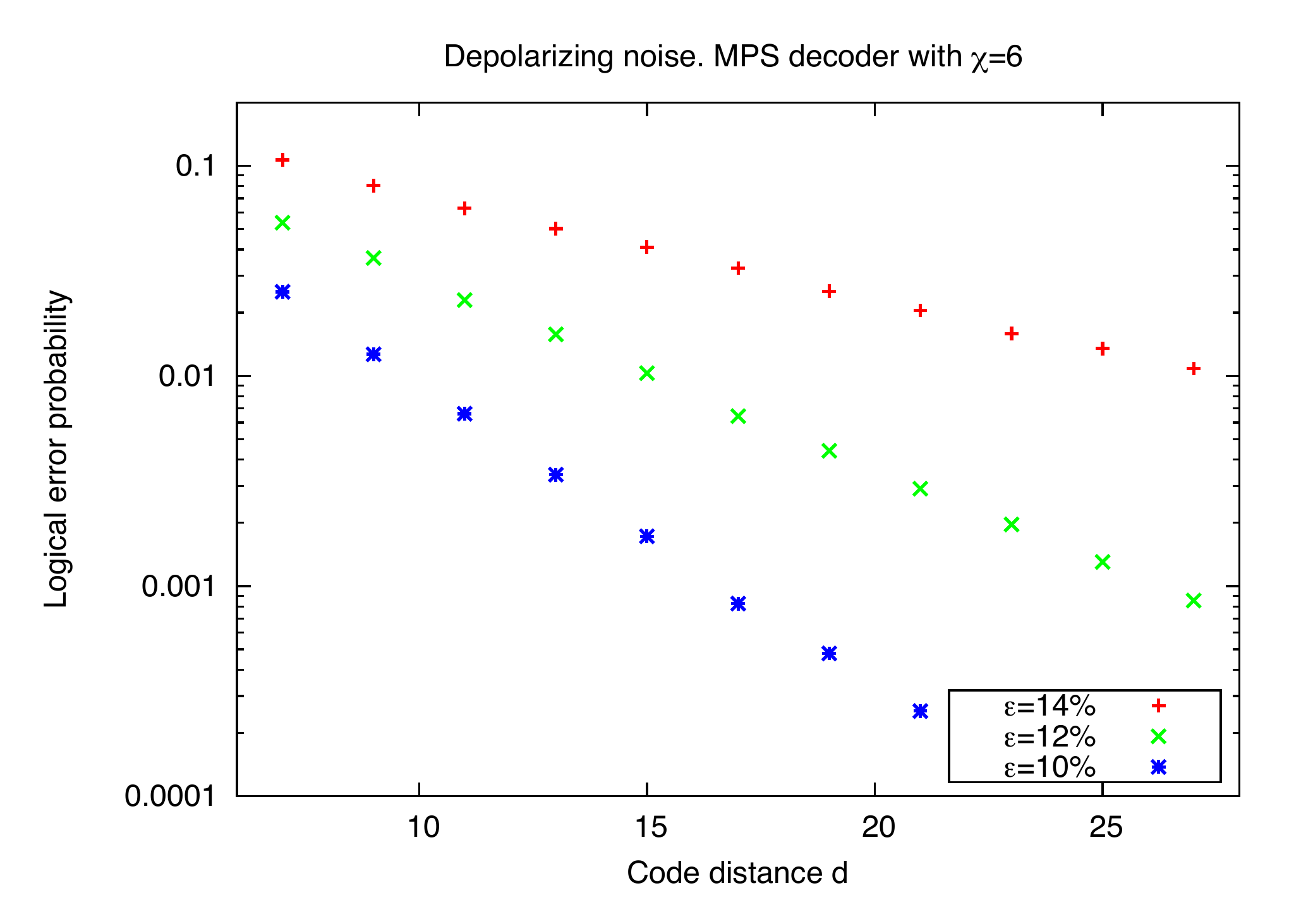}}
\caption{Depolarizing noise: logical error probability of the MPS decoder with $\chi=6$ 
as a function of the code distance $d$ for a 
fixed  error rate $\epsilon=10,12,14\%$.
To compute the logical error
probability, at least $1,000$ failed error correction trials have been accumulated
for each datapoint. 
\label{fig:XYZdecay}
}
\end{figure}

\begin{figure}[h]
\centerline{\includegraphics[height=6.5cm]{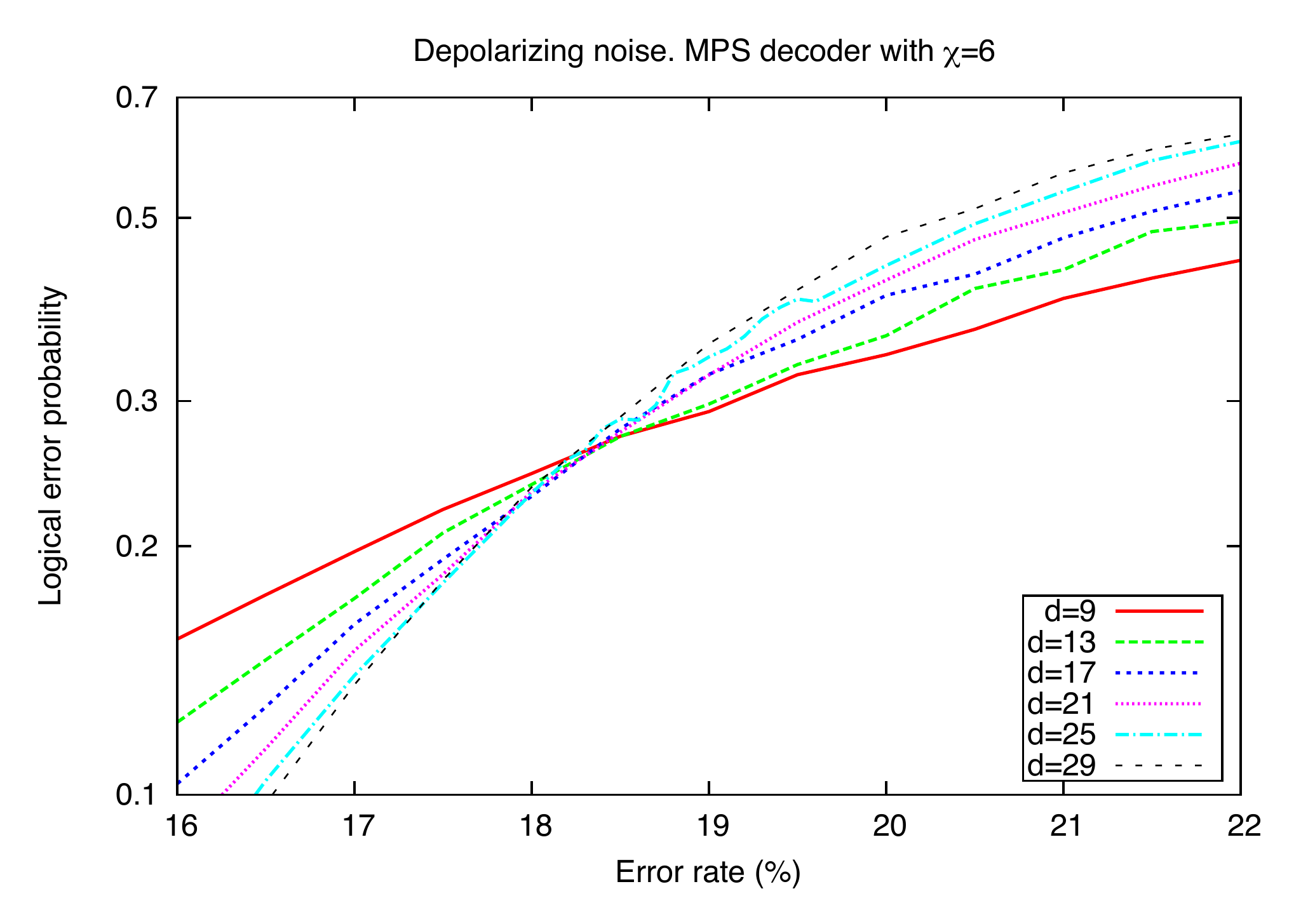}}
\caption{Depolarizing noise: logical error probability of the MPS decoder with $\chi=6$ 
as a function of the error rate $\epsilon$.
Assuming a non-zero error threshold
$\epsilon_0$, the data suggest that $17\%\le \epsilon_0\le 18.5\%$.
To compute the logical error
probability, at least $5,000$ failed error correction trials have been accumulated
for each datapoint. 
\label{fig:XYZth}
}
\end{figure}

The performance of  different decoders  for a fixed code distance $d=25$
and a wide range of error rates is shown at \Fig{XYZcomp}. 
In a striking contrast with the analogous X-noise data, we observed that the MWM decoder
becomes highly non-optimal in the regime of small error rates with the badness parameter above $100$.
This can be attributed to the fact that  MWM decoder often fails to  find the minimum weight error consistent with the syndrome since it  ignores correlations between $X$ and $Z$ errors~\cite{Fowler13optimal}.
We also observed that the logical error probability of MPS decoders converges very quickly as one increases
the bond dimension. The data shown on \Fig{XYZcomp}
indicates that the MPS decoder with $\chi=6$ is nearly optimal for all error rates
and all code distances $d\le 25$.

While the logical error probability is the most natural figure of merit, 
one may also ask how well the MPS-based algorithm with 
a small bond dimension $\chi$ approximates the coset 
probabilities for some fixed syndrome. 
For simplicity, we considered the  trivial syndrome, that is, the cosets 
$\calG,\overline{X}\calG,\overline{Y}\calG$, and $\overline{Z}\calG$.
We observed a very fast convergence for the most likely coset and a  poor
convergence for the remaining cosets, see Tables~\ref{table:cosetX},\ref{table:cosetXYZ}.
 Since the only goal of the decoder is to identify the most
likely coset, the slower convergence for some of  unlikely cosets might not be a serious drawback. 
\begin{table}[h]
\begin{ruledtabular}
\begin{tabular}{c|c|c}
$\chi$ & $\pi(\calG)\cdot 10^{27}$ & $\pi(\bar{X}\calG)\cdot 10^{57}$ \\
\hline
$2$ & $1.78275$ & $4.72777$  \\
$3$ & $1.78277$ & $5.52579$   \\
$4$ & $1.78283$ & $5.80294$  \\
$5$ &  $1.78283$ & $6.03204$  \\
\end{tabular}
\end{ruledtabular}
\caption{X-noise: probabilities of the two cosets  computed by the MPS algorithm.
 The simulation parameters are
$\epsilon=5\%$ and $d=25$. The exact values of the coset probabilities are $\pi(\calG)=1.78283\cdot 10^{-27}$
and $\pi(\overline{X}\calG)=5.58438\cdot 10^{-57}$.}
\label{table:cosetX}
\end{table}

\begin{table}[h]
\begin{ruledtabular}
\begin{tabular}{c|c|c|c|c}
$\chi$ & $\pi(\calG)\cdot 10^{55}$ & $\pi(\bar{X}\calG)\cdot 10^{89}$ & $\pi(\bar{Y}\calG)\cdot 10^{122}$ & 
$\pi(\bar{Z}\calG)\cdot 10^{90}$ \\
\hline
$2$ & $1.11782$ & $2.81823$ & $36.0410$ & $1.64802$ \\
$3$ & $1.11781$ & $2.81777$ & $7.62958$ & $1.70803$  \\
$4$ & $1.11781$ & $2.81781$ & $2.79984$ & $1.78193$ \\
$5$ &   $1.11781$ & $2.81781$ & $3.24487$ & $2.94628$ \\
\end{tabular}
\end{ruledtabular}
\caption{Depolarizing noise: probabilities of the four cosets computed by the MPS algorithm.
The simulation parameters are  $\epsilon=10\%$ and $d=25$.}
\label{table:cosetXYZ}
\end{table}

\onecolumngrid
\begin{center}
\begin{figure}[h]
\centerline{\includegraphics[height=8cm]{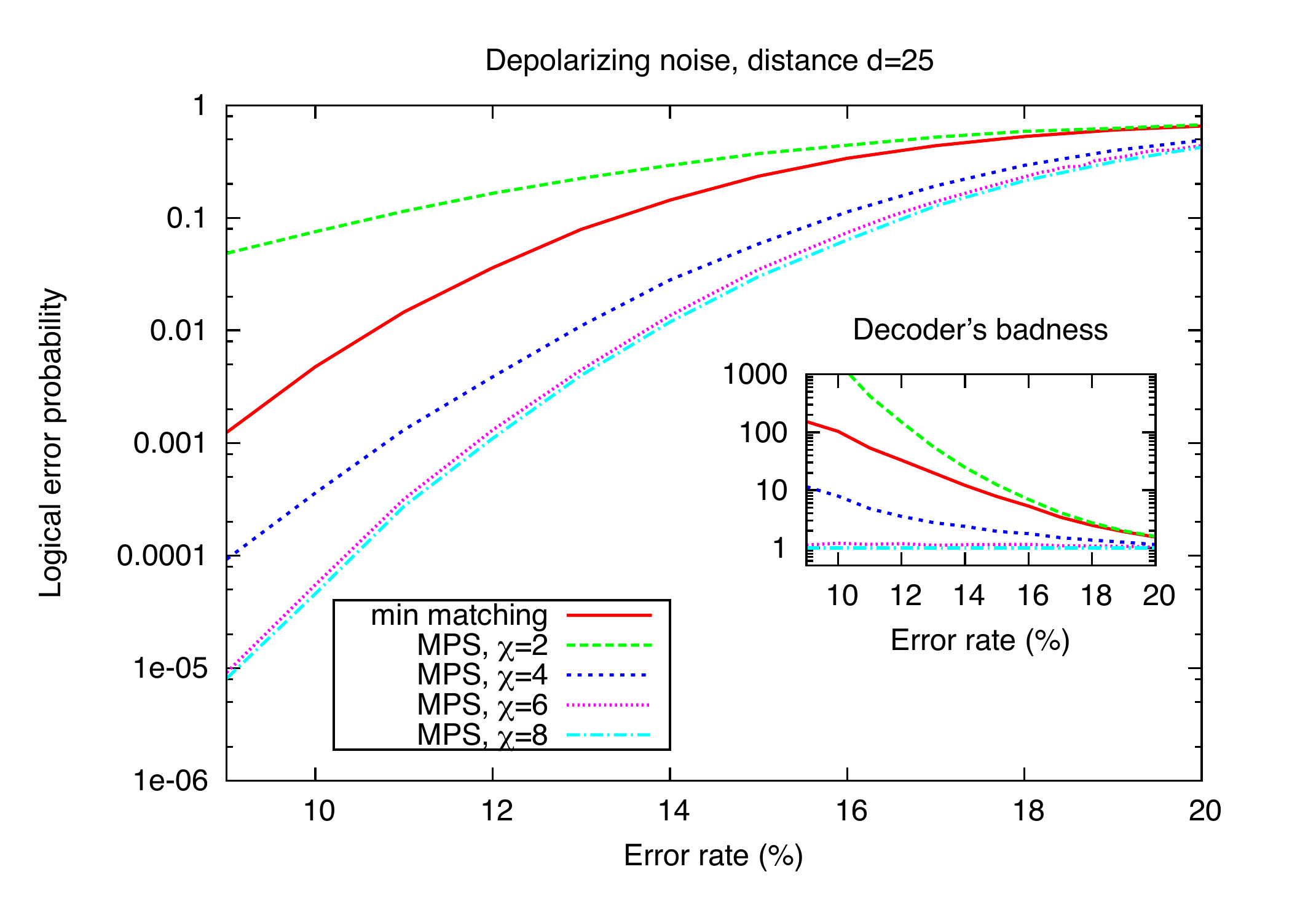}}
\caption{Depolarizing noise: approximate implementations of the ML decoder.
 Logical error probability as a function of the error rate $\epsilon$
is shown.  The red curve represents the  minimum weight matching decoder. 
The inset shows 'badness' of various decoders as a function of the error rate.
We define the badness 
as the ratio between  logical error probabilities of a given decoder and the best decoder (MPS decoder with $\chi=8$).
Each curve has data points at error
rates $\epsilon=9,10,\ldots,20\%$. 
To compute the logical error
probability, at least $1,000$ failed error correction trials have been accumulated
for each datapoint. 
\label{fig:XYZcomp}
}
\end{figure}
\end{center}
\twocolumngrid

The MPS decoder offers a lot of possibilities for improvement. One rather obvious improvement
(employed in the above simulations) is to use a single run of Algorithm~2 to compute two different
coset probabilities. Indeed, suppose we choose the logical operator $\overline{Z}$
supported in the right-most column of the lattice denoted $H^d$ on \Fig{slices}. 
Then the tensor networks constructed for the cosets $\calC^s_I$ and $\calC^s_Z$ are exactly
the same except for the column $H^d$. Since we contract the network column by column starting
from the left-most column $H^1$,  the  difference between the two cosets manifests itself only in the very
last step of  Algorithm~2 (computing the inner product $\langle \hat{H}^d|\psi\rangle$).
Since this step takes a negligible time compared with the rest of the algorithm,
it makes sense to  compute both
probabilities $\pi(\calC^s_I)$ and $\pi(\calC^s_Z)$ by performing a single network contraction. 
The same observation applies to the probabilities $\pi(\calC^s_X)$ and $\pi(\calC^s_Y)$.
We also expect that a choice of the standard error $f(s)$ consistent with the syndrome $s$
may affect the convergence of the algorithm.
 While we have chosen
$f(s)$ by connecting each syndrome to the left/top boundary, it may be advantageous to choose $f(s)$ as
a small-weight error, for example, using the MWM decoder. 
Finally, a challenging open problem is how to extend the MPS decoder to noisy
syndrome extraction. A naive extension would require a contraction of a 3D tensor network.
We anticipate that this problem can be attacked using recently developed algorithms
for simulating 2D quantum systems based on Projected Entangled Pairs States (PEPS),
see~\cite{Murg07,Verstraete04}.

\begin{center}
{\bf Acknowledgments}
\end{center}
We would like to thank Graeme Smith and John Smolin  for helpful comments.
Computational resources for this work were provided by IBM Blue Gene Watson supercomputer center.


\end{document}